\documentclass[twocolumn]{IEEEtran}

\usepackage{amsmath}
\usepackage{amsthm}
\usepackage{color}
\usepackage{soul,xcolor}
\usepackage{graphicx}
\usepackage{amsfonts}
\usepackage{mathtools}

\usepackage{balance}

\newtheorem{theorem}{Theorem}
\newtheorem{definition}{Definition}
\newtheorem{lemma}[theorem]{Lemma}
\newtheorem{corollary}{Corollary}

\newcommand{\diag}{\mathop{\mathrm{diag}}}

\begin{document}
	
\setstcolor{blue}
	
\title{Generalized Compression Strategy for the Downlink Cloud Radio Access Network}
\author{Pratik Patil, Wei Yu,~\IEEEmembership{Fellow,~IEEE}
\thanks{This work is supported by Natural Sciences and Engineering 
Research Council of Canada. The materials in this paper have been presented in
part at the IEEE Communication Theory Workshop (CTW), Nafplio, Greece, May 2016.
The authors are with The Edward S. Rogers Sr. Department of Electrical and
Computer Engineering, University of Toronto, 10 King's College Road, Toronto,
Ontario M5S 3G4, Canada (e-mails: pratik.patil@mail.utoronto.com,
weiyu@ece.utoronto.ca).

Copyright (c) 2017 IEEE. Personal use of this material is permitted.  However, permission to use this material for any other purposes must be obtained from the IEEE by sending a request to pubs-permissions@ieee.org.}
}

\maketitle

\begin{abstract}
This paper studies the downlink of a cloud radio access network (C-RAN) in which a centralized processor (CP) communicates with mobile users through base stations (BSs) that are connected to the CP via finite-capacity fronthaul
links. Information theoretically, the downlink of a C-RAN is modeled as a two-hop broadcast-relay network. Among the various transmission and relaying strategies for such model, this paper focuses on the compression strategy, in
which the CP centrally encodes the signals to be broadcast jointly by the BSs, then compresses and sends these signals to the BSs through the fronthaul links.  We characterize an achievable rate region for a generalized 
compression strategy with Marton's multicoding for broadcasting and multivariate compression for fronthaul transmission. We then compare this rate region with the distributed decode-forward (DDF) scheme, which achieves the
capacity of the general relay networks to within a constant gap, and show that the difference lies in that DDF performs Marton's multicoding and multivariate compression jointly as opposed to successively as in the compression strategy.
A main result of this paper is that under the assumption that the fronthaul links are subject to a \emph{sum} capacity constraint, this difference is immaterial; so,  for the Gaussian network, the compression strategy based on successive encoding can already achieve the capacity region of the C-RAN to within a constant gap, where the gap is independent of the channel parameters and the power constraints at the BSs. As a further result, for C-RAN under individual fronthaul constraints, this paper also establishes that the compression strategy can achieve to within a constant gap to the \emph{sum} capacity.
\end{abstract}


\section{Introduction}

This paper studies the downlink of a cloud radio access network (C-RAN) in which the base stations (BSs) are connected to a centralized cloud-computing-enabled processor through wired or wireless fronthaul links \cite{simeone2016cloud}. Information theoretically, the downlink C-RAN can be modeled as a broadcast-relay channel: the CP broadcasts the user messages to the BSs via the fronthaul links and the BSs act as relays for the mobile users.
This paper considers the C-RAN model where the BSs are connected to the CP
through noiseless digital fronthaul links of finite capacities and there
are no direct links between the CP and the mobile users. In the
ideal case where the capacities of the fronthaul links are infinite, downlink C-RAN model reduces to a multi-antenna broadcast channel.  The optimal transmission strategy in this case is cooperative beamforming combined with dirty-paper coding (DPC) \cite{costa1983writing}. For the practical situation where the fronthaul links have finite capacities, the optimal coding strategy must combine both broadcasting and relaying, and is highly non-trivial; the characterization of the capacity region is still an open problem. This paper makes progress in establishing the achievable rate region of a generalized compression strategy and in showing that it is approximately optimal for the downlink C-RAN under certain conditions.

\subsection{Coding Strategies}

While the C-RAN architecture has been originally motivated by the radio-over-fiber concept \cite{simeone2016cloud}, the information theoretical study of the downlink C-RAN model belongs to that of relay channels, and more specifically relates to the so-called diamond relay channels for which there is an extensive literature, e.g., \cite{schein2001distributed}, \cite{traskov2007reliable}, \cite{kang_diamond}, \cite{chern14}, \cite{kang15}.
In the C-RAN context, there are two main classes of transmission and relaying strategies available in the literature: the data-sharing and the compression strategies.
In the data-sharing strategy, individual user messages are sent directly via
the digital fronthaul to the BSs, which then perform cooperative beamforming to
the users. The capacity constraints of the fronthaul links limit the number of
users whose messages can be sent to each BS, hence limiting the cooperation BS
cluster size for each user. Among the data-sharing schemes, joint encoding
at the BSs can be done using linear beamforming with the sharing of the entire
messages \cite{DaiYu_Access14} or with message splitting \cite{ZakhourGesbert11}.
Generalized versions of the data-sharing strategy using Marton's broadcast coding 
have been proposed for a 2-user 2-BS C-RAN in \cite{LiuKang2014}, and improved upon in
\cite{wang2018, wang2017} by using a common message, and further generalized in
\cite{YiLiu2015} for arbitrary number of users and BSs. Although the data-sharing strategy does not necessarily achieve the capacity in general,
there are some special cases for which it does. For example, 
the achievable rate based on Marton's coding proposed in \cite{BidokhtiKramerShamai2017} for a C-RAN with a single user (but any number of BSs) can be shown to achieve the capacity in some interesting regimes of operation. 
Upper bounds on the sum rate of some other specific cases of C-RAN model are studied in \cite{yang2017upper}. We also mention here that instead of sharing the individual user messages directly, the
CP may send a function of user messages to the BSs. For example, in the reverse
compute-forward strategy \cite{hong_caire_journal}, a function of the messages
is relayed to the BSs using lattice codes. 
As an alternative to the data-sharing strategy, the capacity limitation of the fronthaul links can also be
dealt with using a compression strategy \cite{ParkSimeoneSahinShamai2014}, 
in which the encoding is performed at the CP as a function of the messages of all
users, but in order to accommodate the capacity constraints of the fronthaul links,
the encoded analog signals are compressed and sent to the BSs.  The BSs then
transmit the encoded signals to the users after decompressing the received
compression bits. We note here that a hybrid scheme combining the data sharing
and compression strategies is also possible \cite{patil2018hybrid}.

This paper aims to understand the information theoretical optimality of the 
compression strategy for C-RAN. As pointed out earlier, if the fronthaul
capacity is infinite, the downlink C-RAN reduces to the well-known vector
Gaussian broadcast channel, for which DPC achieves the capacity region.
For the finite fronthaul case, DPC and linear precoding schemes cannot be
applied directly. A compressed version of DPC using independent compression
across the BSs is introduced in \cite{SimeoneSomekhPoorShamai09} and the achievable
user rates are derived for a simplified Wyner type model.
The independent compression scheme can be further improved by using a
multivariate compression strategy across all the BSs
\cite{ParkSimeoneSahinShamai13}.  
The idea is to correlate the quantization noises at the different BSs
to better control the effect of quantization
at the users. The achievable rate
expressions under linear beamforming and multivariate compression for the
Gaussian C-RAN model are given in \cite{ParkSimeoneSahinShamai13} and the
corresponding achievable rate region using dirty paper coding followed by
multivariate compression is given in \cite{simeone2016cloud}.

Can either the data-sharing or compression strategy approach the information
theoretic capacity region of the C-RAN model?  Toward answering this question,
this paper draws inspiration from a new coding strategy named
distributed decode-forward (DDF) \cite{Lim2017DistributedDecodeForward} for
broadcasting multiple messages over a general relay network, which has been
shown to achieve the capacity region of the general Gaussian broadcast relay
network to within a constant gap, which is linear in the number of nodes in the network
but is independent of the channel parameters and the power constraints.
We remark that when specialized to the downlink C-RAN model, the gap can be 
improved from linear to logarithmic in the number of users and BSs
\cite{GangulyKim2017}. Further, it may be possible to further enlarge the rate region of the DDF strategy by incorporating a common codeword, as shown for a two-user two-BS C-RAN model with BS corporation in \cite{wang2018, wang2017}.

\subsection{Contributions}

This paper makes an observation that when specialized to the C-RAN model, 
the DDF strategy resembles the compression strategy for C-RAN, but with a crucial difference that instead of performing the compression followed by Marton's multicoding, the DDF performs both the Marton's coding and multivariate
compression jointly at the CP. As practical implementation for performing successive Marton's coding and multivariate compression would likely be easier, we ask in this paper whether there are conditions under which the difference is immaterial. One of the main results of this paper is that under a sum fronthaul constraint, this is indeed true. Thus, for the Gaussian C-RAN under the sum fronthaul constraint, the compression strategy can already achieve the capacity region to within a constant gap. As a further result, for the Gaussian C-RAN under individual fronthaul constraints, this paper also shows that Marton's encoding followed by multivariate compression can achieve the sum capacity to within a constant gap. More specifically, this paper makes the following contributions:
\begin{enumerate}
\item We provide the achievable rate region of a general form of the
compression strategy that includes Marton's multicoding followed by
multivariate compression for the C-RAN model with digital fronthaul in the
first hop and a general discrete memoryless channel (DMC) in the second hop.  
\item We specialize the DDF strategy to the C-RAN model and compare the
coding strategies of the above generalized compression strategy and the DDF strategy. We
observe that DDF is a further generalization in that the Marton's coding and
multivariate compression are done jointly.
\item We analyze the conditions under which such a generalization of the
compression strategy in the DDF strategy does not strictly enlarge the achievable rate region.
\begin{enumerate}
	\item With any DMC on the second hop, the generalized compression
strategy and the DDF strategy achieve the same rate
region under a sum fronthaul constraint.
	\item With a Gaussian network on the second hop, the sum rate achieved by
the above general compression strategy is within a constant gap to the sum capacity
of C-RAN, where the gap is independent of the network parameters.  
\end{enumerate}
\end{enumerate}

\subsection{Notation and organization}

Random variables are denoted by uppercase letters, their realizations by lowercase letters, and the probability distributions by $p(\cdot)$. Sets are denoted by calligraphic letters, while $[1:n]$ denotes the set $\{1,\ldots,n\}$
for all natural numbers $n$. A subscript for a random variable and its realization denotes its node index. A superscript for a random variable or its realization is a time index that denotes a sequence of random variables or its realizations till that index (e.g., $X_l^n = (X_l^1,\ldots,X_l^n)$ or $x_l^n = (x_l^1,\ldots,x_l^n)$). Random variables can be indexed with sets (e.g., $X(\mathcal{S}) = (X_l : l \in \mathcal{S})$). Bold-face lower case letters are used to denote vectors and bold-face upper case letters are used to denote random vectors or matrices. The standard notations for entropy, $H(X)$, and mutual information, $I(X;Y)$, are used.
Total correlation between a group of random variables is denoted by $T(\cdot)$
and is defined as
\begin{equation}
T(X(\mathcal{S})) = \sum_{l \in \mathcal{S}} H(X_l) - H(X(\mathcal{S})).
\end{equation} 
See \cite{watanabe1960information} for motivation of such a definition and some of its properties. We follow the typicality notation of \cite{ElGamalKim2011NetworkIT} and use $\mathcal{T}^{(n)}_\epsilon$ to denote the set of typical sequences of length $n$ with parameter $\epsilon$.

The rest of the paper is organized as follows. Section \ref{sec:model} provides
a mathematical model for the downlink C-RAN.  Section \ref{sec:compression}
provides the achievable rate region results of the generalized compression
strategy. Section \ref{sec:DDF} specializes the distributed decode-forward
strategy to the downlink C-RAN model under consideration. 
In Section \ref{sec:rateregion_comparison}, we compare the rate regions
achieved by the two strategies and provide conditions under which the two coincide. Section \ref{sec:Conclusion} concludes the paper.

\section{System Model}
\label{sec:model}

\begin{figure}
	\centering
	\includegraphics[width=\columnwidth]{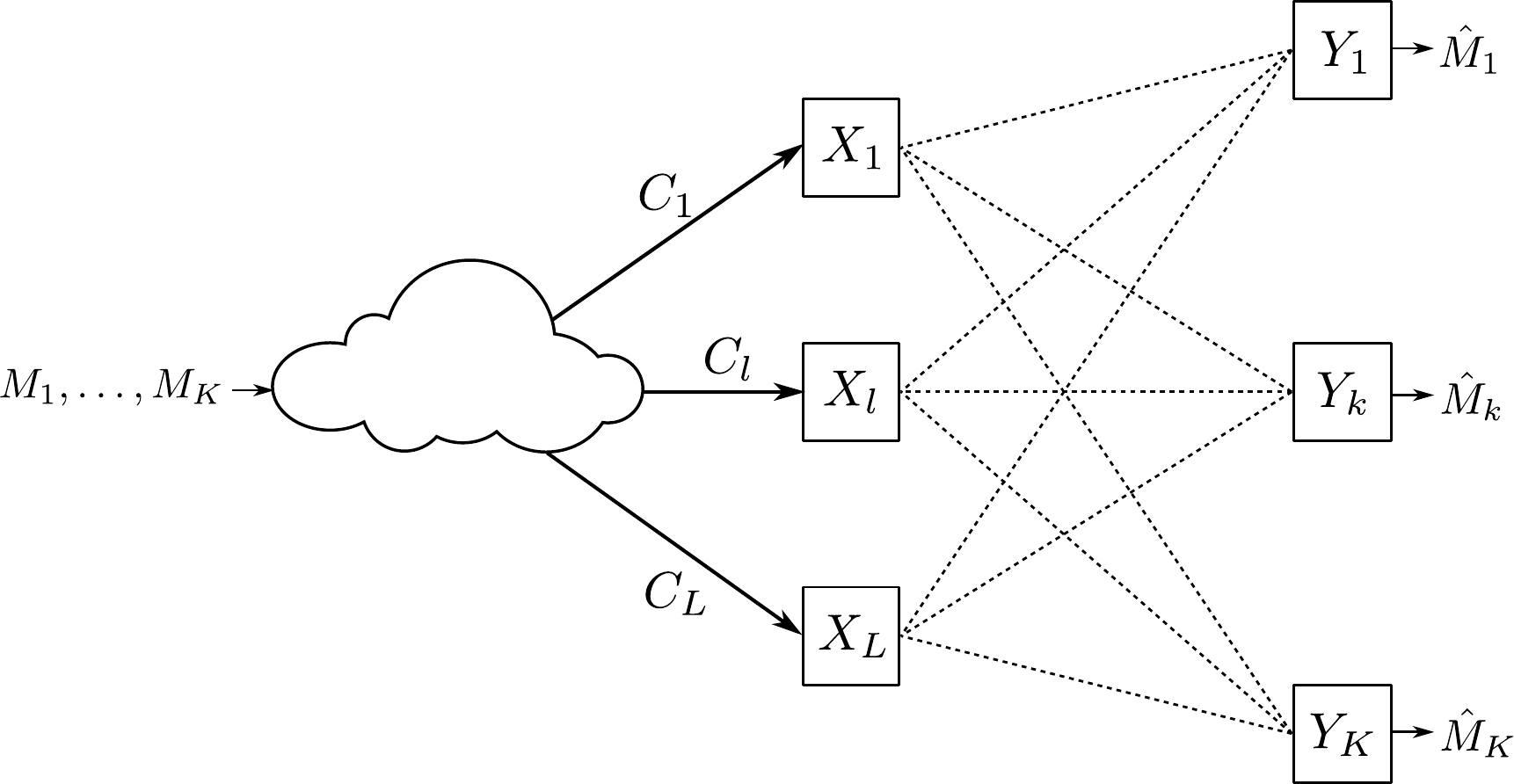}
	\caption{Downlink C-RAN with $L$ BSs, $K$ users, and a channel 
	$p(y_1,\ldots,y_K|x_1,\ldots,x_L)$ between the BSs and the users.} 
	\label{fig:model}
\end{figure}

Consider the downlink of a C-RAN comprising of a CP and $L$ BSs serving $K$ users as shown in Fig.~\ref{fig:model}. 
The CP communicates with BSs
through noiseless fronthaul links of finite capacities, denoted by $C_l$ for BS $l$, $l \in \mathcal{L} \coloneqq [1:L]$. We assume
a discrete memoryless channel $(\mathcal{X}_1 \times \cdots \times
\mathcal{X}_L,p(y_1,\ldots,y_K|x_1,\ldots,x_L),\mathcal{Y}_1 \times \cdots \times \mathcal{Y}_K)$ between 
the BSs and the users.
Let the intended message for user $k$ be denoted by $M_k$, $k \in \mathcal{K} \coloneqq [1:K]$.
A $(2^{nR_1},\ldots,2^{nR_K},n)$ code for the downlink C-RAN consists of a mapping at the CP from the $K$ user messages $(m_1,\ldots,m_K)
\in [1:2^{nR_1}] \times \cdots \times [1:2^{nR_K}]$ to $L$ indices
$(t_1,\ldots,t_L) \in [1:2^{nC_1}] \times \cdots \times [1:2^{nC_L}]$, encoders at the $L$ BSs that map the index $t_l$ to a codeword $x_l^n(t_l)$, and decoders at the $K$ users that estimate $\hat{m}_k$ based on the received signals $y_k^n$. The average probability of error is defined as $P_e^{(n)} = P\{\hat{m}_k \ne
m_k \text{ for some } k \in \mathcal{K}\}$. A rate tuple $(R_1,\ldots,R_K)$ is
achievable if there exists a sequence of codes such that $\lim_{n \rightarrow
\infty} P_e^{(n)} = 0$.

Of particular interest is the special case where the channel between the BSs
and the users is a Gaussian channel such that
\begin{equation}
\label{eq:Gaussian_distribution}
\mathbf{Y} = \mathbf{H} \mathbf{X} + \mathbf{Z},
\end{equation}
where $\mathbf{Y} = \left[Y_1,\ldots,Y_K\right]^T$ are the received signals at the $K$
users, $\mathbf{X} = \left[X_1,\ldots,X_L\right]^T$ are the transmitted signals from the $L$ BSs, $\mathbf{H} = [\mathbf{h}_1,\ldots,\mathbf{h}_k]^T$ is the $K \times L$ channel matrix consisting of channel vectors $\mathbf{h}_1$ to $\mathbf{h}_k$ for users 1 to $K$, respectively, and $\mathbf{Z}=[Z_1,\ldots,Z_K]^T \sim \mathcal{N}(0,\sigma^2 \mathbf{I})$ 
is the additive white Gaussian noise. We assume all the BSs have an average power constraint of $P$ without loss of generality.
For simplicity, both the BSs and the users are assumed to be equipped with a single antenna in this paper. 

\section{Generalized Compression Strategy}

\label{sec:compression}

The compression strategy has been extensively studied in the literature 
\cite{ParkSimeoneSahinShamai2014, SimeoneSomekhPoorShamai09,
ParkSimeoneSahinShamai13}.  The coding strategy involves two steps. First, the
CP jointly encodes the user messages. Second, the encoded signals are
compressed in order to accommodate them through the fronthaul links. Different
options for joint encoding include linear beamforming strategies such as
zero-forcing or regularized zero-forcing, or non-linear beamforming strategy
such as dirty paper coding. Different options for compression include
independent compression or multivariate compression.
The main point of this section is to show that these specific compression
strategies previously studied in \cite{ParkSimeoneSahinShamai2014, SimeoneSomekhPoorShamai09,
ParkSimeoneSahinShamai13} are special forms of a generalized compression
strategy in which joint encoding is performed via Marton's multicoding.
The coding strategy proposed in this paper does not, however, incorporate the possibility of a common codeword, as done in \cite{wang2018, wang2017}.
\begin{theorem}
\label{thm:1}
	A rate tuple $(R_1,\ldots,R_K)$ is achievable for the downlink C-RAN using 
the compression strategy with Marton's multicoding followed by multivariate compression if
	\begin{equation}
	\sum_{k \in \mathcal{D}} R_k < \sum_{k \in \mathcal{D}} I(U_k;Y_k) - T(U(\mathcal{D})) \label{eq:Marton}
	\end{equation}
	for all $\mathcal{D} \subseteq \mathcal{K}$ such that
	\begin{equation}
	\sum_{l \in \mathcal{S}} C_l > I(U(\mathcal{K});X(\mathcal{S})) + T(X(\mathcal{S})) \label{eq:multivariate_compression}
	\end{equation}
	for all $\mathcal{S} \subseteq \mathcal{L}$ for some distribution $p(u_1,\ldots,u_K,x_1,\ldots,x_L)$.
\end{theorem}

The proof of achievibility is in Appendix \ref{appendix:proof_compression}. The set of inequalities \eqref{eq:Marton} represents the achievable user
rates using Marton's multicoding for broadcast channels. In linear beamforming,
the $U$'s are just the messages and are thus independent of each other.  The
advantage of using Marton's multicoding is to introduce correlation among $U$'s
for the possibility of increased rates. But doing so incurs a penalty that
depends on the total correlation present among $U$'s. 
DPC is an example of such Marton's coding.

One way to implement Marton' coding is through successive encoding of user
messages. Assuming without loss of generality that the encoding order is user
$1,\ldots,K$. The achievable rate for user 1 is $I(U_1;Y_1)$. Treating user
1's message as known interference, user 2 achieves a rate of 
$I(U_2;Y_2) - I(U_1;U_2)$; and user $k$ achieves a rate of
$I(U_k;Y_k)-I(U_k;U_{k-1},\ldots,U_1)$. We remark that, as pointed out in
\cite{zhang2007successive}, there is a subtle issue that such successive
encoding may not achieve the entire Marton's region. The reason is that 
even though the set function $\sum_{k \in \mathcal{D}} I(U_k;Y_k) - T(U(\mathcal{D}))$ satisfies the submodular property, it is not guaranteed that it satisfies the monotone property that the successive user rates
$I(U_k;Y_k) - I(U_k;U_{k-1},\ldots,U_1)$ are always non-negative. 
Hence, the Marton's region itself is not guaranteed to be a polymatroid.  
The rest of this paper ignores this subtlty and assumes that
the Marton' rate region is polymatroid so that we can use successive encoding
to achieve the corner points of the rate region.

The set of inequalities \eqref{eq:multivariate_compression} represents the multivariate compression of
$U(\mathcal{K})$ into $X$'s that are transmitted by the BSs. If the BSs were
co-located and can cooperate, the amount of quantization needed for compression is
simply the first term $I(U(\mathcal{K});X(\mathcal{S}))$. If the BSs are
distributed and cannot cooperate, there is a penalty in terms of the correlation
between the signals transmitted by the BSs.

Similar to the successive encoding for the Marton's region, the multivariate
compression can also be implemented in a successive manner
\cite{ParkSimeoneSahinShamai13}. Without loss of generality, let's assume that the
encoding order is BS $1,\ldots,L$. The fronthaul required to compress the
signal for BS 1 is $I(U(\mathcal{K});X_1)$. After compressing the signal
for BS 1, the fronthaul required to compress BS 2's signal is given by
$I(U(\mathcal{K});X_2) + I(X_2;X_1)$; and for any BS $l$ the
fronthaul required is $I(U(\mathcal{K});X_l) +
I(X_l;X_{l-1},\ldots,X_1)$. It can be verified that the fronthaul region in
general is a contra-polymatroid \cite{zhang2007successive}.

The above achievability region has been presented at \cite{Wei_CTW} and is 
subsequently generalized in \cite{wang2018} to the case with
common information and BS cooperation where there are two BSs in the C-RAN.
We now specialize the generalized compression strategy for Gaussian C-RAN \eqref{eq:Gaussian_distribution} using various choices for the distribution $p(u_1,\ldots,u_K,x_1,\ldots,x_L)$ and show how it results in the known compression strategies in the literature. We assume 2 BSs and 2 users for simplicity.

Consider the strategy of linear beamforming followed by compression. In this case, we choose the messages $U$'s as independent Gaussian random variables, compute the beamformed signals to be transmitted by the BSs at the cloud, then compress using either independent compression or multivariate compression.  Mathematically, we express the
distribution $p(u_1,u_2,x_1,x_2)$ as
\begin{equation}
\label{eq:beamforming_compression}
\mathbf{X} = \mathbf{W} \mathbf{U} + \mathbf{N},
\end{equation}
where $\mathbf{W} = [\mathbf{w}_1, \mathbf{w}_2]$ is a beamforming
matrix with beamformers $\mathbf{w}_1$ and $\mathbf{w}_2$ for users 1 and 2,
respectively, and $\mathbf{N}$ is the quantization noise, assumed to be
a Gaussian vector $\mathcal{N}(0, \mathbf{Q})$. Here, $\mathbf{U}=[U_1,U_2]^T
\sim \mathcal{N}(0,\mathbf{I})$ are the independent message signals for the two
users.
The achievable user rates of the generalized compression strategy with this
choice of $\mathbf{U}$ are given by
\begin{align}
R_1^{\rm{linear}} &= I(U_1;Y_1) \\
&= \frac{1}{2} \log \left( 1 + \frac{\mathbf{h}_1^T \mathbf{w}_1
\mathbf{w}_1^T \mathbf{h}_1}{\mathbf{h}_1^T \mathbf{w}_2 \mathbf{w}_2^T
\mathbf{h}_1 + \mathbf{h}_1^T \mathbf{Q} \mathbf{h}_1 + \sigma^2}  \right)
\label{eq:R1}
\end{align}
for user $1$, and similarly for user 2
\begin{align}
R_2^{\rm{linear}} &= I(U_2;Y_2) \\
&= \frac{1}{2} \log \left( 1 + \frac{\mathbf{h}_2^T \mathbf{w}_2
\mathbf{w}_2^T \mathbf{h}_2}{\mathbf{h}_2^T \mathbf{w}_1 \mathbf{w}_1^T
\mathbf{h}_2 + \mathbf{h}_2^T \mathbf{Q} \mathbf{h}_2 + \sigma^2}  \right).
\label{eq:R2}
\end{align}
Note that the covariance matrix of $\mathbf{N}$ enters the rate expression 
as an additional noise term. Depending on the compression strategy used,
$\mathbf{Q}$ is either diagonal in case of independent compression owing to
independent noise components among the compressed BS signals, or a full matrix
in case of multivariate compression, due to the introduced correlation among
the noise components of $\mathbf{N}$.
In the independent compression case, let
$\mathbf{Q}=\diag(q_{11},q_{22})$, and $\mathbf{w}_1=[w_{11},w_{12}]^T$
and $\mathbf{w}_2=[w_{21},w_{22}]^T$. 
The amount of fronthaul needed to support compression at BS 1 is
\begin{align}
C_1^{\rm{linear,indep}} &= I(X_1;U_1,U_2) \\
&= \frac{1}{2} \log \left(1 + \frac{w_{11}^2+w_{21}^2}{q_{11}}\right),
\label{eq:C1_independent_compression}
\end{align}
and similarly for BS 2,
\begin{align}
C_2^{\rm{linear,indep}} &= I(X_2;U_1,U_2) \\ 
&= \frac{1}{2} \log \left(1 +  \frac{w_{12}^2+w_{22}^2}{q_{22}}\right).
\label{eq:C2_independent_compression} 
\end{align}
For the multivariate compression, the required fronthaul rates $(C_1,C_2)$ must
be inside a rate region. 
A corner point of the region assuming a successive compression strategy with the 
order of compression to be BS 1 followed by BS 2 is as following.
For BS 1, $C_1^{\rm{linear,multi}}$ is exactly the same as the independent compression case (\ref{eq:C1_independent_compression}),
but for BS 2, we have
\begin{align}
C_2^{\rm{linear,multi}} &= I(X_2;{U}_1,{U}_2|X_1) + I(X_1;X_2) \\
&= I(X_2;{U}_1,{U}_2) + I(X_1;X_2|{U}_1,{U}_2)\\
&= \frac{1}{2} \log \left( 1 + \frac{w_{12}^2+w_{22}^2}{q_{22}} \right)
\nonumber \\
& \qquad
+  \frac{1}{2} \log \left( \frac{q_{22}}{q_{22} - q_{21} q_{11}^{-1}q_{12}}\right)
\end{align}
where $\mathbf{Q} = \left[\begin{array}{cc} q_{11} & q_{12} \\  q_{21} & q_{22} \end{array} \right]$ is a full matrix 
whose correlation structure, although leading to higher $(C_1,C_2)$,
nevertheless allows possible reduction in the effective noise in 
the achievable rates (\ref{eq:R1}) and (\ref{eq:R2}), thereby potentially providing an overall
benefit. These derived rate expressions can be shown to be equivalent to that
in \cite{ParkSimeoneSahinShamai13}.

The linear beamforming strategy can be improved by
introducing correlation between the $U$'s. One example of using such
correlation is DPC, which is capacity achieving for the Gaussian vector
broadcast channel (i.e., with infinite $C_1$ and $C_2$). With DPC, the $U$'s
are now random vectors. Although using the $U$'s designed for the broadcast
channel for C-RAN is not necessarily optimal when $C_1$ and $C_2$ are finite,
it is nevertheless instructive to write down the rate expressions to gain some
insight.
Assume an ordering of DPC with user 1 followed by user 2.  The auxiliary random variables for DPC can be constructed as follows.  Let $\mathbf{S}_1$ and $\mathbf{S}_2$ be two independent Gaussian vectors with covariance matrices $\mathbf{K}_1$ and $\mathbf{K}_2$.  Fix $\mathbf{N} \sim \mathcal{N}(0, \mathbf{Q})$.
We choose
\begin{align}
\mathbf{U}_1 = \mathbf{S}_1, \mathbf{U}_2 = \mathbf{S}_2 + \mathbf{A} \mathbf{S}_1, \mathbf{X} = \mathbf{S}_1 + \mathbf{S}_2 + \mathbf{N},
\end{align}
where
$\mathbf{A} = \mathbf{K}_2 \mathbf{h}_2 \left( \mathbf{h}_2^T \left( \mathbf{K}_2 + \mathbf{Q} \right)
\mathbf{h}_2 + \sigma^2 \right)^{-1} \mathbf{h}_2^T$ and $\mathbf{N}$ is the quantization noise.
This choice of the auxiliary variables allows the interference from user 1 to
be completely pre-subtracted from user 2 \cite{color_paper}, resulting in the following 
achievable user rates 
\begin{align}
	R_1^{\rm{DPC}} &= I(\mathbf{U}_1;Y_1) \\
	&= \frac{1}{2} \log \left(1 +
\frac{\mathbf{h}_1\mathbf{K}_1\mathbf{h}_1^T}{\mathbf{h}_1\mathbf{K}_2\mathbf{h}_1^T
+ \mathbf{h}_1\mathbf{Q}\mathbf{h}_1^T + \sigma^2 }\right)
\end{align}
for user $1$ who sees user $2$ as noise, 
and
\begin{align}
	R_2^{\rm{DPC}} &= I(\mathbf{U}_2;Y_2) - I(\mathbf{U}_1;\mathbf{U}_2) \\
	&= I(\mathbf{X};Y_2 | \mathbf{S}_1) \\
	&= \frac{1}{2} \log \left( 1 +
\frac{\mathbf{h}_2\mathbf{K}_2\mathbf{h}_2^T}{\mathbf{h}_2\mathbf{Q}\mathbf{h}_2^T + \sigma^2}\right)
\end{align}
for user 2, who no longer sees user 1 as interference.
The required fronthaul rates depend on whether independent or multivariate
compression is performed.
Let the covariance matrix of $\mathbf{S}_1 + \mathbf{S}_2$ be $\mathbf{K}_1 +
\mathbf{K}_2 = \left[\begin{array}{cc} \Sigma_{11} & \Sigma_{12} \\
\Sigma_{21} & \Sigma_{22} \end{array} \right]$ and the covariance matrix of the quantization noise $\mathbf{N}$ be $\mathbf{Q} = \left[\begin{array}{cc} q_{11} & q_{12} \\  q_{21} & q_{22} \end{array} \right]$. 
With independent compression, we have $q_{12}=q_{21}=0$ and 
\begin{align}
\label{eq:C1_DPC_indep}
C_1^{\rm{DPC,indep}} &= I(X_1;\mathbf{U}_1,\mathbf{U}_2) \\ 
&= \frac{1}{2} \log \left(1+ \frac{ \Sigma_{11} }{{q}_{11}} \right)
\end{align}
\begin{align}
\label{eq:C2_DPC_indep}
C_2^{\rm{DPC,indep}} &= I(X_2;\mathbf{U}_1,\mathbf{U}_2) \\
&= \frac{1}{2} \log \left(1+ \frac{ \Sigma_{22} }{{q}_{22}} \right). 
\end{align}
For multivariate compression, assuming the corner point of the fronthaul rate
region with the ordering of compression to be BS 1
followed by BS 2, we have $C_1^{\rm{DPC,multi}}$ for BS 1 exactly the same as in the case of
independent compression (\ref{eq:C1_DPC_indep}),
but for BS 2, we need additional fronthaul capacity given by 
\begin{align}
	C_2^{\rm{DPC,multi}} &= I(X_2;\mathbf{U}_1,\mathbf{U}_2|X_1) + I(X_1;X_2) \\
	    &= I(X_2;\mathbf{U}_1,\mathbf{U}_2) + I(X_1;X_2|\mathbf{U}_1,\mathbf{U}_2)
		\\
	    &= \frac{1}{2} \log \left(1+ \frac{ \Sigma_{22} }{{q}_{22}} \right) \nonumber \\
& \qquad + \frac{1}{2} \log \left( \frac{{q}_{22}}{{q}_{22} - {q}_{21} {q}_{11}^{-1}{q}_{12}}\right).
\end{align}
These rate expressions for DPC over C-RAN are equivalent to the ones given in
\cite{simeone2016cloud}. They can be interpreted as the compression of $\mathbf{X}$
at the CP for transmission to the BSs. The above more rigorous derivation is based
on transmitting $\mathbf{U}$ to the BSs via compression.

\section{Distributed Decode-Forward}
\label{sec:DDF}

The main objective of this paper is to understand whether the generalized
compression strategy can approximately achieve the capacity region of the Gaussian C-RAN
model. Toward this end, we examine the DDF strategy 
\cite{Lim2017DistributedDecodeForward}, which is a general coding scheme for
broadcasting multiple messages over a general relay network that combines
Marton's coding for the broadcast channel with partial decode-forward for the
relay channel. The coding scheme involves using auxiliary random variables at
each node in the network that implicitly carry information about the user messages.
By specializing the DDF strategy to the C-RAN setup, we write down a succinct form of the achievable rate region using DDF and a simplified coding strategy that can be readily compared with the generalized compression strategy.

\begin{theorem}[\cite{Lim2017DistributedDecodeForward}]
\label{thm:2}
A rate tuple $(R_1,\ldots,R_K)$ is achievable for the downlink C-RAN using the DDF strategy if
	\begin{align}
	\sum_{k \in \mathcal{D}} R_k & <  \sum_{k \in \mathcal{D}} I(U_k;Y_k) + \sum_{l \in \mathcal{S}} C_l - T(U(\mathcal{D}),X(\mathcal{S}))\\
	& = \sum_{k \in \mathcal{D}} I(U_k;Y_k) - T(U(\mathcal{D})) \nonumber \\
	& \qquad + \sum_{l \in \mathcal{S}} C_l - I(U(\mathcal{D});X(\mathcal{S})) - T(X(\mathcal{S})) \label{eq:DDF_rate}
	\end{align}
	for all $\mathcal{D} \subseteq \mathcal{K}$ and $\mathcal{S} \subseteq
\mathcal{L}$ for some distribution $p(u_1,\ldots,u_K,x_1,\ldots,x_L)$.
\end{theorem}

The proof of achievibility is in Appendix \ref{appendix:proof_ddf}.
Comparing the DDF coding strategy of Theorem \ref{thm:2} with that of the generalized compression strategy of Theorem \ref{thm:1}, we observe that the DDF strategy generalizes the compression strategy by combining Marton's multicoding with multivariate compression and jointly encoding the Marton's and compression codewords. The key difference is that, in the compression strategy, Marton's codewords are formed first, then the multivariate compression codewords are computed in a sequential order. Note that the rate region in Theorem \ref{thm:1} is in general a subset of the rate region in Theorem \ref{thm:2} as any distribution $p(u_1,\ldots,u_K,x_1,\ldots,x_L)$ satisfying the multivariate compression constraints \eqref{eq:multivariate_compression} results in generalized compression rates \eqref{eq:Marton} which are also achievable in the form \eqref{eq:DDF_rate} using the DDF strategy.

A key advantage of enlarging the allowable distributions to beyond the ones that explicitly satisfy the fronthaul constraints is that it permits a proof of the result that the DDF strategy can achieve to within a constant gap to the cut-set bound of the general Gaussian broadcast relay channel \cite{Lim2017DistributedDecodeForward}. The ingenious choice of $p(u_1,\ldots,u_K|x_1,\ldots,x_L)$ proposed in \cite{Lim2017DistributedDecodeForward} that accomplishes this task is a
distribution for $p(u_1,\ldots,u_K|x_1,\ldots,x_L)$  that tries to mimic the Gaussian channel distribution $p(y_1,\ldots,y_K|x_1,\ldots,x_L)$. 
We now specialize the result of \cite{Lim2017DistributedDecodeForward} to the C-RAN setup \eqref{eq:Gaussian_distribution}. The DDF strategy can be shown to achieve to within a constant gap to the cut-set outer bound by choosing $\mathbf{X}$ to be a vector
of $L$ independent Gaussian random variables $\mathcal{N}(0,P)$ and by choosing 
\begin{equation}
\label{eq:constant_gap_distribution}
	\mathbf{U} = \mathbf{H} \mathbf{X} + \tilde{\mathbf{Z}},
\end{equation}
where $\tilde{\mathbf{Z}} \sim \mathcal{N}(0,\sigma^2 I)$ is independent of $\mathbf{Z}$.
With this choice of $p(u_1,\ldots,u_K|x_1,\ldots,x_L)$, we have 

\begin{corollary}[\cite{GangulyKim2017}] 
\label{cor:individual_fronthaul}
With Gaussian $p(y_1,\ldots,y_K|x_1,\ldots,x_L)$ on the second hop of the C-RAN
model and individual fronthaul constraints $(C_1,\ldots,C_K)$, the DDF strategy
achieves a rate region within a constant gap to the capacity region of
C-RAN, where the gap is independent of the channel, the BS power constraints, and the fronthaul constraints, and only depends on the number of BSs and users. 
\end{corollary}

A natural question at this point is whether we can use the generalized
compression strategy to accomplish the same. The next section gives some
partial answers in the affirmative but under specific conditions.

\section{Compression versus DDF}
\label{sec:rateregion_comparison}

DDF generalizes the compression strategy, so the achievable rate region of the generalized compression strategy is a subset of the DDF region in general. This section asks the question of whether this subset inclusion is strict. The main result here is that, under certain conditions, the rate regions of the two strategies actually coincide. Specifically, we show that under a sum fronthaul constraint, the rate regions of the two strategies coincide for any discrete memory channel on the second hop of C-RAN. In other words, under a sum fronthaul constraint, performing Marton's coding and multivariate compression separately does not reduce the achievable user rates.  As a second result of this section, we show that in the special case of Gaussian networks but under individual fronthaul constraint, the compression strategy achieves the sum capacity of C-RAN to within a constant gap.  These results are useful, because successive Marton's coding and multivariate compression is likely easier to implement than the joint encoding for DDF. For example, an architecture based on successive estimation of minimum mean-squared error and per-BS compression to achieve the multivariate compression region is proposed in \cite{ParkSimeoneSahinShamai2014}, 
while polar coding based scheme to achieve the general Marton's region for a 2-user broadcast channel is proposed in \cite{mondelli2015achieving}.

\subsection{Rate Region Under Sum Fronthaul Constraint}

\begin{definition}
Consider the closure of the convex hull of achievable rate-fronthaul tuples
$(R_1,\ldots,R_K,C_1,\ldots,C_L)$ using the generalized compression strategy
satisfying (\ref{eq:Marton})-(\ref{eq:multivariate_compression}) over all joint
distributions $p(u_1,\ldots,u_K,x_1,\ldots,x_L)$ satisfying possibly input
constraints on $(x_1,\ldots,x_L)$.  Define
$\mathcal{R}^{\textrm{s}}_{\textrm{COM}}(C)$ to be the projection of the above
set along a sum fronthaul constraint $C$, i.e., the set of rate tuples
$(R_1,\ldots,R_K)$ such that $C_l \ge 0$ and $\sum_{l} C_l \le C$.
\end{definition}

\begin{definition}
Consider the closure of the convex hull of achievable rate-fronthaul tuples
$(R_1,\ldots,R_K,C_1,\ldots,C_L)$ using the DDF strategy satisfying
(\ref{eq:DDF_rate}) over all joint distributions
$p(u_1,\ldots,u_K,x_1,\ldots,x_L)$ satisfying possibly input
constraints on $(x_1,\ldots,x_L)$.  Define
$\mathcal{R}^{\textrm{s}}_{\textrm{DDF}}(C)$ to be the projection of the above
set along a sum fronthaul constraint $C$, i.e., the set of rate tuples
$(R_1,\ldots,R_K)$ such that $C_l \ge 0$ and $\sum_{l} C_l \le C$.
\end{definition}

Let us write down the two rate regions
$\mathcal{R}^{\textrm{s}}_{\textrm{COM}}(C)$ and
$\mathcal{R}^{\textrm{s}}_{\textrm{DDF}}(C)$ defined above more explicitly. 
For the compression strategy,
under a fixed joint distribution $p(u_1,\ldots,u_K,x_1,\ldots,x_L)$ and a fixed sum fronthaul constraint $C$, only the constraint for $\mathcal{S} = \mathcal{L}$ is active in \eqref{eq:multivariate_compression}. Therefore, 
the set of $(R_1,\ldots,R_K)$ that satisfies the sum fronthaul constraint $C$
under a fixed joint distribution $p(u_1,\ldots,u_K,x_1,\ldots,x_L)$ is 
described by the constraints
\begin{align}
\sum_{k \in \mathcal{D}} R_k & < \sum_{k \in \mathcal{D}} I(U_k;Y_k) - T(U(\mathcal{D})) \label{eq:COM_sumfronthaul} \\
C & >  I(U(\mathcal{K});X(\mathcal{L})) + T(X(\mathcal{L})) \label{eq:sum_fronthaul_compression}
\end{align}
over all $\mathcal{D} \subseteq \mathcal{K}$. 
The set $\mathcal{R}^{\textrm{s}}_{\textrm{COM}}(C)$ is then the projection of
the closure of the convex hull of these $(R_1,\ldots,R_K,C)$ tuples, where the convex hull
is taken over both the distributions as well as $C$.

Similarly, for the DDF strategy, under a fixed distribution
$p(u_1,\ldots,u_K,x_1,\ldots,x_L)$ and a fixed sum fronthaul constraint $C$, the
active constraints are those corresponding to $\mathcal{S} = \emptyset$ for the case when the sum fronthaul is large enough to accommodate the compression of all BS signals (i.e., $C > I(U(\mathcal{D});X(\mathcal{L})) + T(X(\mathcal{L}))$), which corresponds to the Marton's region, or $\mathcal{S} = \mathcal{L}$ for the case when the sum fronthaul is not large enough to accommodate the compression of all BS signals; see \cite{zhou2016optimal} for a similar result. The set of $(R_1,\ldots,R_K)$ that satisfies
the sum fronthaul constraint $C$ under a fixed joint distribution
$p(u_1,\ldots,u_K,x_1,\ldots,x_L)$ is thus described by the constraints
\begin{align}
\sum_{k \in \mathcal{D}} R_k & < \sum_{k \in \mathcal{D}} I(U_k;Y_k) - T(U(\mathcal{D})) \label{eq:DDF_sumfronthaul_1} \\
\sum_{k \in \mathcal{D}} R_k & < \sum_{k \in \mathcal{D}} I(U_k;Y_k) + C - T(U(\mathcal{D}),X(\mathcal{L})) \label{eq:DDF_sumfronthaul_2} \\
&= \sum_{k \in \mathcal{D}} I(U_k;Y_k) - T(U(\mathcal{D})) \nonumber \\
& \qquad + C - I(U(\mathcal{D});X(\mathcal{L})) - T(X(\mathcal{L})) \label{eq:DDF_sumfronthaul_3}
\end{align}
over all $\mathcal{D} \subseteq \mathcal{K}$. 
The set $\mathcal{R}^{\textrm{s}}_{\textrm{DDF}}(C)$ is then the projection of the closure of the 
the convex hull of the above $(R_1,\ldots,R_K,C)$ tuples, where the convex hull (i.e.,  time-sharing) 
is taken over both the distributions as well as $C$.

\begin{theorem}
	\label{thm:sum_fronthaul}
	For the downlink C-RAN with a general DMC $p(y_1,\ldots,y_K|x_1,\ldots,x_L)$ in the second
hop and a sum fronthaul constraint $C$, we have $\mathcal{R}^{\textrm{s}}_{\textrm{COM}}(C) = \mathcal{R}^{\textrm{s}}_{\textrm{DDF}}(C)$.
\end{theorem}

We briefly explain the main ideas of the proof using an illustrative 2-BS 2-user example. Consider any given channel $p(y_1,y_2|x_1,x_2)$ in the second hop of C-RAN and a sum fronthaul constraint $C$.  The rate region using the
generalized compression strategy is given by
\begin{align}
R_1 & < I(U_1;Y_1) \\
R_2 & < I(U_2;Y_2) \\
R_1 + R_2 & < I(U_1;Y_1) + I(U_2;Y_2) - I(U_1;U_2)
\end{align}
under joint distributions $p(u_1,u_2,x_1,x_2)$ that satisfy
$
I(U_1,U_2;X_1,X_2) + I(X_1;X_2) < C.
$
For the DDF strategy, the rate region under any fixed distribution 
$p(u_1,u_2,x_1,x_2)$ can be expressed as
\begin{align}
R_1 &< I(U_1,Y_1) \nonumber \\
& + \min \left\{ \begin{aligned} & 0, \\
								 & C - I(U_1;X_1,X_2) - I(X_1;X_2)
				 \end{aligned} \right\} \label{eq:sumfronthaulregion_1} \\
R_2 & < I(U_2,Y_2) \nonumber \\
& + \min \left\{ \begin{aligned} & 0, \\
								 & C - I(U_2;X_1,X_2) - I(X_1;X_2)
				 \end{aligned} \right\} \\
R_1 + R_2 & < I(U_1,Y_1) + I(U_2,Y_2) - I(U_1;U_2) \nonumber \\
& + \min \left\{ \begin{aligned} & 0, \\
					             & C - I(U_1,U_2;X_1,X_2) - I(X_1;X_2)
				\end{aligned} \right\}. \label{eq:sumfronthaulregion_2}
\end{align}
To show that the generalized compression and DDF regions coincide (after convex
hull), we start with the DDF region under some fixed distribution 
$p(u_1,u_2,x_1,x_2)$.
If the sum fronthaul capacity is such that $C > I(U_1,U_2;X_1,X_2) +
I(X_1;X_2)$ under a fixed distribution $p(u_1,u_2,x_1,x_2)$, then both rate
regions are exactly the same. The interesting case is when the distribution
$p(u_1,u_2,x_1,x_2)$ is such that 
$
C < I(U_1,U_2;X_1,X_2) + I(X_1;X_2),
$
which is allowed under the DDF strategy but not under the generalized compression strategy. But, we show that by time-sharing across varying $p(u_1,u_2,x_1,x_2)$, (specifically, the original
$p(u_1,u_2,x_1,x_2)$ and one with either of the users shut off), the DDF achievable rate region can nevertheless be achieved using time-sharing of the generalized compression strategies while satisfying an average fronthaul constraint. Intuitively, the penalty that the DDF strategy pays to go beyond the fronthaul capacity is at least as large as the penalty for shutting off the appropriate users.
The proof for the general case of arbitrary number of users and BSs makes use of the polymatroidal structure of the rate region to characterize all the corner points of the rate region achieved by the DDF strategy and constructs appropriate time-shared compression strategies to achieve all such corner points. The full proof is relegated to Appendix \ref{appendix:sum_fronthaul}.

Since the DDF strategy is known to achieve the rate
region of the C-RAN to within a constant gap for the Gaussian network, 
having the generalized compression rate region coincide with the DDF region under the sum fronthaul constraint immediately gives us the following corollary.

\begin{corollary}
\label{cor:sum_fronthaul}
With Gaussian $p(y_1,\ldots,y_K|x_1,\ldots,x_L)$ on the second hop of the C-RAN
model and under a sum fronthaul constraint $C$, the compression strategy
achieves a rate region to within a constant gap to the capacity region. The gap is independent of the channel, the BS power constraints, and the sum fronthaul constraint, and
only depends on the number of BSs and users. 
\end{corollary}

As a remark, we wonder whether the generalized compression and DDF rate regions coincide not just under the sum fronthaul constraint, but also individual fronthaul constraints. While the answer to this question is not yet clear, we note here that the successive coding strategy of computing Marton's codewords $(U_1,\ldots,U_K)$ first, then forming the compression codewords $(X_1,\ldots,X_L)$ is not the only way to perform
successive encoding.  There is also the possibility of breaking the encoding
into more than two steps. As an example, consider a 2-BS 2-user C-RAN. The
compression encoding order that we consider in this paper encodes $(U_1,U_2)$
jointly first, and then $(X_1,X_2)$ is computed. But it is possible to encode
$(U_1,X_1)$ first, and then encode $(U_2,X_2)$. Such a re-ordering can potentially help
user 2 because knowing the exact signals to be transmitted to user 1 can
benefit the search for $U_2$ to align its correlation with $U_1$ to
appropriately cancel the interference at user 2. Thus, interleaving in the
encoding of $U$'s and $X$'s is likely needed in order to achieve the same rate
region as DDF under arbitrary fronthaul constraints.  
However, as shown in the next section, if we only consider the sum rate, the
two-step encoding of the generalized compression strategy indeed achieves the
sum capacity of C-RAN to within a constant gap, even under individual fronthaul
constraints, if we assume a Gaussian channel $p(y_1,\ldots,y_K|x_1,\ldots,x_L)$
and use a Gaussian $p(u_1,\ldots,u_K,x_1,\ldots,x_L)$ in the encoding.

It is worth pointing out that similar results exist for the
uplink C-RAN. In the uplink, by comparing the joint decoding of quantized BS
signals and user messages using noisy network coding versus the successive
decoding of quantized signals followed by decoding of user messages,
it is possible to establish that the successive decoding of
quantized signals and user messages achieves the same sum rate as the noisy
network coding strategy (see \cite{zhou2016optimal}, and also \cite{aguerri2017capacity} 
under an ``oblivious'' assumption), while successive decoding (that allows for interleaving within
successive quantized signal decoding and user message decoding) achieves the
same rate region as noisy network coding under a sum fronthaul constraint
\cite{zhou2016optimal}. Just as in the downlink, it
is still an open question as to whether successive decoding can match the noisy
network coding rate region under arbitrary individual fronthaul constraints by considering all possible interleaving combinations across quantization and user messages, In fact, there is a duality between uplink and downlink C-RAN. It can be shown under the assumption of independent compression that the uplink and downlink compression strategies achieve exactly the same rate region for the 
C-RAN model \cite{liu2016uplink}; a similar result is expected to hold 
under the multivariate compression. This suggests an even
stronger connection between the generalized compression strategies in the
uplink and the downlink C-RAN.

\subsection{Sum Rate Under Individual Fronthaul Constraints}
\label{sec:Gaussian}

In this section, we consider the general case of individual fronthaul constraints instead of
restricting to the sum fronthaul constraint as in the previous section.  However, we focus
on the sum rate only, and aim to find the approximate sum capacity of C-RAN
under arbitrary fronthaul constraints.  The main result of this section is that
under a Gaussian C-RAN model, the generalized compression strategy can achieve
a sum rate which is within a constant gap to the cut-set bound of C-RAN under
individual fronthaul constraints. 
More precisely, consider the Gaussian C-RAN model specified in
\eqref{eq:Gaussian_distribution}. For the Gaussian channel, recall that if we set 
the distribution $p(u_1,\dots,u_K,x_1,\dots,x_L)$ according to
\eqref{eq:constant_gap_distribution}, the DDF strategy can be shown to 
achieve to within a constant gap to the capacity region of the Gaussian C-RAN. 
For convenience, we call the distribution in 
\eqref{eq:constant_gap_distribution} the \emph{constant-gap distribution}. 
We show in this section that the sum rate achieved by the DDF strategy for the
Gaussian C-RAN under the constant-gap distribution can also be achieved using
the generalized compression strategy under the same set of fronthaul constraints. 

For each fixed distribution, we can write down the achievable sum rate of the
DDF and the generalized compression strategies explicitly.  
The sum rate achieved by the DDF strategy is given by $R$ that satisfies
\begin{equation}
R < \sum_{k \in \mathcal{K}} I(U_k;Y_k) + \sum_{l \in \mathcal{S}} C_l - T(U(\mathcal{K}),X(\mathcal{S})), \label{eq:sumrate_DDF_1}
\end{equation}
for all $\mathcal{S} \subseteq \mathcal{L}$ under some distribution $p(u_1,\dots,u_K,x_1,\dots,x_L)$.
The sum rate achieved by the generalized compression strategy is given by
$R$ that satisfies
\begin{align}
R & < \sum_{k \in \mathcal{K}} I(U_k;Y_k) - T(U(\mathcal{K})) \label{eq:sumrate_COM_1}\\
\sum_{l \in \mathcal{S}} C_l  & >  I(U(\mathcal{K});X(\mathcal{S})) - T(X(\mathcal{S})) \label{eq:sumrate_COM_2},
\end{align}
for all $\mathcal{S} \subseteq \mathcal{L}$ under some joint distribution $p(u_1,\dots,u_K,x_1,\dots,x_L)$.

\begin{definition}
\label{def:DDF_sumrate}
Consider the closure of the convex hull of achievable sum-rate-fronthaul tuples 
$(R,C_1,\ldots,C_L)$ for the C-RAN with the Gaussian channel model
(\ref{eq:Gaussian_distribution}) using the DDF strategy as expressed in
(\ref{eq:sumrate_DDF_1}) under the constant-gap distribution
(\ref{eq:constant_gap_distribution}) with the BS powers constrained by the power constraint $P$.
Define $R^\text{g}_{\text{DDF}}$ to be the maximum sum rate 
under individual fronthaul constraints $(C_1,\ldots,C_L)$ in this set. 
\end{definition}

\begin{definition}
Consider the closure of the convex hull of achievable sum-rate-fronthaul tuples
$(R,C_1,\ldots,C_L)$ of the C-RAN with the Gaussian channel model
(\ref{eq:Gaussian_distribution}) using the generalized compression strategy 
(\ref{eq:sumrate_COM_1})-(\ref{eq:sumrate_COM_2}) under the constant 
under the constant-gap distribution (\ref{eq:constant_gap_distribution}) with the BS powers constrained by the power constraint $P$.
Define $R^\text{g}_{\text{COM}}$ to be the maximum sum rate
under individual fronthaul constraints $(C_1,\ldots,C_L)$ in this set.  
\end{definition}

Comparing the sum rate of DDF in (\ref{eq:sumrate_DDF_1}) with the sum rate of
generalized compression in (\ref{eq:sumrate_COM_1})-(\ref{eq:sumrate_COM_2}), we
clearly have $R^\text{g}_{\text{COM}} \le R^\text{g}_{\text{DDF}}$. 
We show in this section that actually
$R^\text{g}_{\text{COM}} = R^\text{g}_{\text{DDF}}$. 
As a consequence, we have the following main theorem of this section. 

\begin{theorem}
\label{thm:sum_rate}
For the downlink C-RAN with a memoryless Gaussian channel on the second hop, the compression
scheme achieves a sum rate to within a constant gap to the cut-set bound under
individual fronthaul constraints $(C_1,\dots,C_L)$. The gap is independent
of the channel parameters, the BS power constraints, and the individual fronthaul constraints, and only depends on
the number of BSs and users.
\end{theorem}

We briefly explain the key ideas of the proof again using the illustrative 2-BS 2-user example. Under a fixed distribution $p(u_1,u_2,x_1,x_2)$, the sum rate achieved by the
DDF strategy is given by
\begin{align}
\label{eq:sumrate_ddfregion}
R & < I(U_1;Y_1) + I(U_2;Y_2) - I(U_1;U_2) \nonumber \\
& + \min \left\{ \begin{aligned} & 0, \\
									& C_1 - I(U_1,U_2;X_1), \\
									& C_2 - I(U_1,U_2;X_2), \\
									& C_1 + C_2 - I(U_1,U_2;X_1,X_2) - I(X_1;X_2)
\end{aligned} \right\}.
\end{align}
Likewise, for the generalized compression strategy, the sum rate is given by
\begin{align}
\label{eq:sumrate_comregion}
R & < I(U_1;Y_1) + I(U_2;Y_2) - I(U_1;U_2) \\
C_1 & >  I(U_1,U_2;X_1) \label{eq:COM_2BS_1} \\
C_2 & >  I(U_1,U_2;X_2) \label{eq:COM_2BS_2} \\
C_1 + C_2 & >  I(U_1,U_2;X_1,X_2) + I(X_1;X_2).  \label{eq:COM_2BS_3}
\end{align}
Clearly, if the fronthaul capacity constraints $C_1$ and $C_2$ are such that
under the distribution $p(u_1,u_2,x_1,x_2)$, the fronthaul constraints
(\ref{eq:COM_2BS_1})-(\ref{eq:COM_2BS_3}) for the compression strategy are all
satisfied, then the sum rate for the compression strategy is exactly equal to that of the DDF strategy. 
However, if either $C_1$ or $C_2$ or both are not large enough so that
some of the fronthaul constraints are violated, then the DDF strategy can still
provide an achievable rate-tuple, but the sum rate would be smaller than
$I(U_1;Y_1) + I(U_2;Y_2) - I(U_1;U_2)$ by a penalty term equal to how large the
maximum violation in the three fronthaul constraints is. For the compression strategy, however, whenever the fronthaul constraints are not
satisfied, we can no longer use the distribution $p(u_1,u_2,x_1,x_2)$ directly.
The idea of the proof is that we can modify the distribution (specifically, by time-sharing between the original $p(u_1,u_2,x_1,x_2)$ and that with one of the users turned off), so
that under the new distribution, we stay within the allowed fronthaul
constraint and achieve a sum rate $I(U_1;Y_1) + I(U_2;Y_2) - I(U_1;U_2)$ that
is at least as large as the penalized sum rate of the DDF strategy.
The proof for the general case of arbitrary number of users
and BSs uses the contra-polymatroidal structure of the fronthaul region to characterize the corner points for each fixed sum rate $R$ under the DDF strategy.
Using appropriate time-sharing schemes in the generalized compression strategy, we show that each such corner point is achievable
in the compression strategy with a sum rate at least as large as $R$. The complete proof is relegated to Appendix~\ref{appendix:sum_rate}.  

As a final remark, we mention the work of \cite{GangulyKim2017}, which shows
that the gap between the achievable rate region and the cut-set bound for DDF
can be refined, so that the gap is logarithmic in the number of BSs and users,
instead of being linear as in \cite{Lim2017DistributedDecodeForward}. The
refinement uses a slightly modified form of the constant-gap distribution. 
The equivalence result shown in this section also works for this modified
constant-gap distribution. Thus, a similar refinement can be used to conclude that the compression
strategy can achieve the sum capacity of the C-RAN network to within a constant gap which
is logarithmic in the number of BSs and users. A different improvement in the gap for the DDF strategy is proved in \cite{wang2018}, and is also applicable to our result.

\subsection{Sum Rate Under Sum Fronthaul Constraint}

The previous two sections show that even though the DDF strategy allows
for distributions $p(u_1,\ldots,u_K,x_1,\ldots,x_L)$ that can compress beyond
the fronthaul constraints, under certain conditions, the compression strategy,
which compresses within the fronthaul constraints, can achieve the same
rate-region or the same sum rate if we allow time-sharing between different
achievable rate tuples of the compression strategy. Applying this result to the
Gaussian C-RAN gives us the conclusion that time-sharing of compression
strategies can achieve the sum capacity of Gaussian C-RAN to within a constant gap. 
This section provides a slightly stronger statement. We show that for maximizing
the \emph{sum} rate under the \emph{sum} fronthaul constraint, there exists a
Gaussian compression strategy that achieves to within a constant gap of the
cut-set bound even \emph{without} time-sharing. 

Let us first write down the achievable sum rates for the DDF and compression
strategies under a sum fronthaul constraint $C$. 
From \eqref{eq:DDF_sumfronthaul_1} and \eqref{eq:DDF_sumfronthaul_3}, we have that
the achievable sum rate $R^s_{\text{DDF}}$ for the DDF strategy under the sum
fronthaul constraint is
\begin{align}
	R^s_{\text{DDF}} & < \sum_{k \in \mathcal{K}} I(U_k;Y_k) - T(U(\mathcal{K})) \nonumber \\
	& \quad + \min \left\{ \begin{aligned} & 0, \\
													& C - I(U(\mathcal{K});X(\mathcal{L})) - T(X(\mathcal{L}))
									\end{aligned} \right\},
	\label{eq:sumrate_sumfronthaul_DDF}
\end{align}
for some distribution $p(u_1,\ldots,u_K,x_1,\ldots,x_L)$.
Similarly, from \eqref{eq:COM_sumfronthaul}, the achievable sum rate $R^s_{\text{COM}}$ using the compression strategy under the sum fronthaul constraint is given by
\begin{equation}
	R^s_{\text{COM}} < \sum_{k \in \mathcal{K}} I(U_k;Y_k) - T(U(\mathcal{K})),
\end{equation}
for some distribution  $p(u_1,\ldots,u_K,x_1,\ldots,x_L)$ that satisfies
\begin{equation}
	C > I(U(\mathcal{K});X(\mathcal{L})) + T(X(\mathcal{L})) \label{eq:sumrate_sumfronthaul_compression}.
\end{equation}
Consider the channel model \eqref{eq:Gaussian_distribution}. We know that
the DDF strategy can achieve to within a constant gap to the capacity region (and
hence the sum capacity) of this Gaussian C-RAN model under individual fronthaul
constraints (and hence also the sum fronthaul constraint) by using the distribution
given by \eqref{eq:constant_gap_distribution}. We now show that by using a
(possibly) modified version of this distribution, the compression strategy can
achieve the same sum rate under the sum fronthaul constraint.

We consider two cases. If under the distribution given in \eqref{eq:constant_gap_distribution}, we have
\begin{align}
\label{eq:sum_rate_sum_fronthaul_case2}
	C &\ge I(U(\mathcal{K});X(\mathcal{L})) + T(X(\mathcal{L})) \\
	  &= \frac{1}{2} \log | \mathbf{I} + P \mathbf{H} \mathbf{H}^T|,
\end{align}
then we can simply use the same distribution in the compression strategy to
achieve the same rate. 
If $C < \frac{1}{2} \log | \mathbf{I} +  P \mathbf{H} \mathbf{H}^T|$, we propose to modify the distribution in
\eqref{eq:constant_gap_distribution} in such a way that when used in the
compression strategy, the fronthaul constraint is satisfied, and further, it
achieves a higher sum rate than the DDF strategy. The proposed modification is to reduce 
the power of $X$'s by a factor $\gamma < 1$. We find $\gamma$ such that
\begin{equation}
	C = \frac{1}{2} \log | \mathbf{I} +  \gamma P \mathbf{H} \mathbf{H}^T|.
\end{equation}
This allows us to compress with the same sum fronthaul rate as the DDF strategy. 
To show that compression with this modified distribution actually improves upon
the DDF sum rate, we compare the sum rate achieved by the compression
strategy to that with DDF under the modified distribution as follows:
\begin{align}
R^s_\text{COM} &= \sum_{k \in \mathcal{K}} I(U'_k;Y'_k) - T(U'(\mathcal{K})) \label{eq:sumrate_com_start}\\
&\overset{(a)}{=} I(U'(\mathcal{K});X'(\mathcal{L})) - \sum_{k \in \mathcal{K}} I(U'_k;X'(\mathcal{L})|Y'_k) \\	
&= C - \sum_{k \in \mathcal{K}} \frac{1}{2}\log\left(1 + \frac{\sum_{l=1}^L h_{k,l}^2 \gamma P}{\sum_{l=1}^L h_{k,l}^2 \gamma P + \sigma^2}\right) \\
&> C - \sum_{k \in \mathcal{K}} \frac{1}{2}\log\left(1 + \frac{\sum_{l=1}^L h_{k,l}^2 P}{\sum_{l=1}^L h_{k,l}^2 P + \sigma^2}\right) \\
&\overset{(b)}{\ge} I(U(\mathcal{K});X(\mathcal{L})) - \sum_{k \in \mathcal{K}} I(U_k;X(\mathcal{L})|Y_k) \\ 
&\overset{(c)}{=} R^s_\text{DDF} \label{eq:sumrate_ddf_finish}
\end{align}
where $(a)$ and $(c)$ are due to Lemma \ref{lemma:equi_sumrate} in Appendiex
\ref{appendix:sum_rate}, and $(b)$ is due to (\ref{eq:sum_rate_sum_fronthaul_case2}). 
Since the DDF sum rate achieves to within a constant gap to the sum capacity of
C-RAN, the above shows that this choice of (non-time-shared) distribution 
also achieves to within a constant gap to the sum capacity.

\subsection{Numerical Example}

We provide a numerical example that illustrates a performance
comparison between the compression and DDF strategies, along with the more
traditional beamforming based strategy that takes fronthaul into account.

Consider the Gaussian channel \eqref{eq:Gaussian_distribution} with 2 BSs and 2 users. As a baseline, we consider beamforming followed by compression, where we use the zero-forcing beamformers in the directions of $\mathbf{W} = \mathbf{H}^{-1}$. We then allocate powers across the two normalized beams $\mathbf{w}_1$ and $\mathbf{w}_2$, and also find the quantization noise levels
to maximize the sum rate by exhaustive search, while satisfying the sum
fronthaul constraint. For the DDF strategy, the sum rate is calculated as given
in \eqref{eq:sumrate_sumfronthaul_DDF}. For the generalized compression scheme,
we construct the explicit distribution that achieves the sum capacity to within a constant gap as explained in the previous section. The parameter $\gamma$ is found using a line search between $[0,1]$. 

\begin{figure}
	\centering
	\includegraphics[width=\columnwidth]{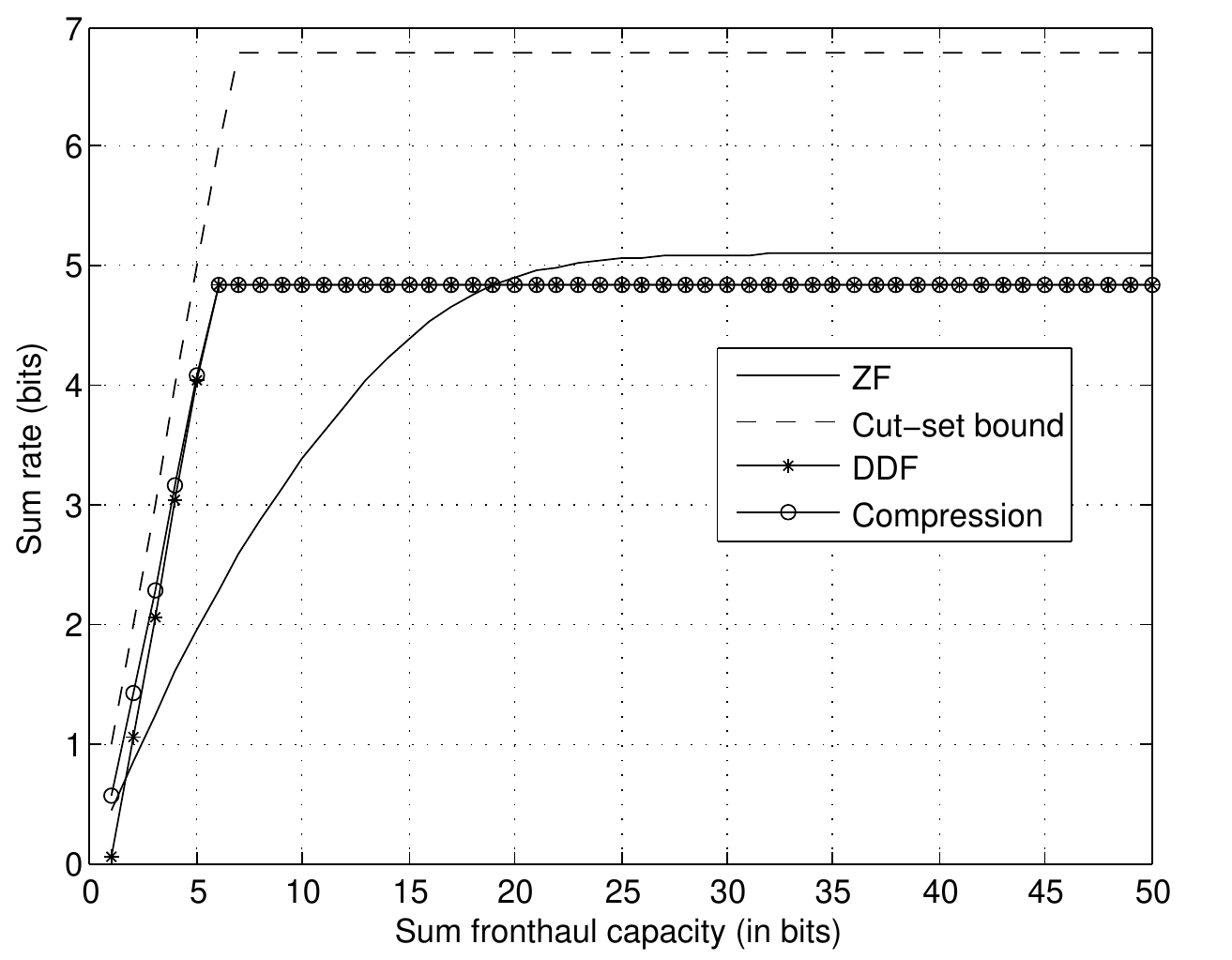}
	\caption{Sum rate comparison of the zero-forcing, DDF, and compression strategies, along with the cut-set outer bound. Note that the DDF and generalized compression strategies are evaluated under their respective distributions that achieve to within a constant gap to capacity for each strategy. In particular, they are not evaluated under their respective optimal distributions, so the numerical result does not imply that the compression strategy outperforms DDF.
	} 
	\label{fig:sumrate_comparison}
\end{figure}

Fig.~\ref{fig:sumrate_comparison} shows the sum rate achieved using these
strategies as a function of the sum fronthaul capacity available for a fixed
real-valued channel (generated at random according to a Rayleigh fading
distribution).  Individual BS power constraints of $P=100$ and background
noise $\sigma^2 = 1$ are assumed. For comparison, we also plot the cut-set
bound. The figure shows that the compression strategy performs nearly the same
as the DDF strategy for most of the sum fronthaul capacity range. It performs
slightly better than the DDF strategy at very low sum fronthaul capacities
because of the improvement in the gap as a result of choosing $\gamma < 1$ to
accommodate the fronthaul capacity as shown from \eqref{eq:sumrate_com_start} to
\eqref{eq:sumrate_ddf_finish}. We note that both the compression and the DDF
strategies are within a constant gap to the cut-set bound. When $C <
\frac{1}{2} \log | \mathbf{I} + P \mathbf{H}  \mathbf{H}^T|$, the
gap is at most 1 bit, while for higher values of $C$, the gap is at most 2 bits. 
As compared to the zero-forcing strategy, we observe that the generalized
compression strategy performs much better at lower sum fronthaul capacities,
because the zero-forcing beam direction does not account for the quantization
noise. At higher sum fronthaul capacities, all three strategies saturate. However, the zero-forcing strategy may achieve a higher sum rate than the compression and DDF strategies, because the
latter does not explicitly null interference, but only aims to provide a
universal strategy that approximately achieves the cut-set bound for all values of the sum fronthaul constraint. The choice of $\mathbf{U} = \mathbf{H} \mathbf{X} +
\tilde{\mathbf{Z}}$ for the channel $\mathbf{Y} = \mathbf{H} \mathbf{X} +
{\mathbf{Z}}$ is in a sense trying to invert the channel, but not exactly. 

Even though the generalized compression strategy achieves to within a constant gap under certain conditions, it is important to note that the gap depends on the number of users and BSs, at least logarithmically. Therefore, the data-sharing strategies might perform better than the generalized compression strategy, especially in the low-power regime or when the channel matrix is ill-conditioned; see \cite{wang2018} and \cite{patil2018hybrid} for some numerical evidence along these lines.

\section{Conclusion}
\label{sec:Conclusion}

This paper investigates the compression strategy for the downlink of a C-RAN
from an information theoretic point of view. The paper first generalizes the
existing compression strategies to include Marton's multicoding followed by
multivariate compression, then analyzes the resulting rate region for a C-RAN
with a general DMC between the BSs and the users. When
compared with the DDF strategy specialized to the downlink C-RAN, it is pointed
out that DDF is a generalization of the compression strategy where the Marton's
multicoding and the multivariate compression are done jointly as opposed to successively in the compression strategy. The paper then shows that under
a sum fronthaul constraint, such generalization does not lead to higher rates
and the rate regions of the two strategies coincide. Thus, for the Gaussian C-RAN under a sum fronthaul constraint, the compression strategy already achieves the capacity region to within a constant gap. Furthermore, for the Gaussian C-RAN under individual fronthaul constraints, the paper shows that the two-phase compression strategy can achieve a sum rate that is within a constant gap to the cut-set bound. These results provide a justification for the practical choice of the two-phase compression strategy for the downlink C-RAN.

\begin{appendices}

\section{Proof of Theorem \ref{thm:1}}
\label{appendix:proof_compression}

	We provide a proof sketch by first describing the coding scheme, then
	establishing the conditions on the achievable rates for vanishing probability 
	of error $P_e^{(n)}$.
	
	Fix the distribution $p(u_1,\ldots,u_K,x_1,\ldots,x_L)$. Let $\epsilon > 0$.
	\begin{enumerate}
		\item Codebook generation:
		We generate a random codebook for Marton's multicoding
		according to $u_k^n(m_k,l_k) \sim \prod_{i=1}^{n} p_{U_k}(u_{ki})$ for
		$(m_k,l_k) \in [1:2^{nR_k}] \times [1:2^{n\tilde{R}_k}]$, $k \in \mathcal{K}$. 
		Similarly, we generate a random codebook for multivariate compression according
		to $x_l^n(t_l) \sim \prod_{i=1}^n p_{X_l}(x_{li})$ for $t_l \in [1:2^{nC_l}]$,
		$l \in \mathcal{L}$.
		\item Encoding at the CP:
		To send $[m_1:m_K]$, we find $[l_1:l_K]$ such that $[u_1^n(m_1,l_1):u_K^n(m_K,l_K)] \in \mathcal{T}_\epsilon^{(n)}$. Then, we find $[t_1:t_L]$ such that $[u_1^n(m_1,l_1):u_K^n(m_K,l_K),x_1^n(t_1):x_L^n(t_L)] \in \mathcal{T}_\epsilon^{(n)}$. Finally, we forward $t_l$ to BS $l$.
		\item Mapping at the BSs: BSs transmit $x_l^n(t_l)$ to users.
		\item Decoding at the users:
		User $k$ finds $(\hat{m}_k,\hat{l}_k)$ such that $(u_k^n(\hat{m}_k,\hat{l}_k),y_k^n) \in \mathcal{T}_\epsilon^{(n)}$.
	\end{enumerate}
	
	In order to show that the average probability of error $P_e^{(n)}$ for the
	coding scheme vanishes as $n \to \infty$, we analyze three sources of error. 
	For encoding at the central processor, we can find the incides $[l_1:l_K]$ 
	correctly with high probability if $ \sum_{k \in \mathcal{D}} \tilde{R}_k > T(U(\mathcal{D}))$ due to the multivariate covering lemma \cite[Lemma 14.1]{ElGamalKim2011NetworkIT}.
	Similarly, we can find the indices $[t_1:t_L]$ correctly with high probability if 
	$\sum_{l \in \mathcal{S}} C_l > I(U(\mathcal{K});X(\mathcal{S})) + T(X(\mathcal{S}))$.
	Finally, the decoding at the user side is successful with high probability if $R_k + \tilde{R}_k <
	I(U_k;Y_k)$ due to the joint typicality lemma \cite[p. 29]{ElGamalKim2011NetworkIT}. Using the Fourier-Motzkin elimination, we project
	out the auxiliary rates $\tilde{R}_k$ to obtain required the rate region.

\section{Proof of Theorem \ref{thm:2}}
\label{appendix:proof_ddf}

	We specialize the DDF coding scheme to the C-RAN model as follows. 
	
	Fix the distribution $p(u_1,\ldots,u_K,x_1,\ldots,x_L)$. Let $\epsilon > 0$.
	\begin{enumerate}
		\item Codebook generation:
		We generate a random codebook for Marton's multicoding
		according to $u_k^n(m_k,l_k) \sim \prod_{i=1}^{n} p_{U_k}(u_{ki})$ for
		$(m_k,l_k) \in [1:2^{nR_k}] \times [1:2^{n\tilde{R}_k}]$, $k \in \mathcal{K}$. 
		Similarly, we
		generate a random codebook for multivariate compression according to 
		$x_l^n(t_l) \sim \prod_{i=1}^n p_{X_l}(x_{li})$ for $t_l \in [1:2^{nC_l}]$, $l \in \mathcal{L}$.
		
		\item Encoding at the CP:
		To send $[m_1:m_K]$, we find $[l_1:l_K,t_1:t_L]$ such that $[u_1^n(m_1,l_1):u_K^n(m_K,l_K),x_1^n(t_1):x_L^n(t_L)] \in \mathcal{T}_\epsilon^{(n)}$.
		
		\item Mapping at the BSs: BSs transmit $x_l^n(t_l)$ to users.
		
		\item Decoding at the users:
		User $k$ finds $(\hat{m}_k,\hat{l}_k)$ such that $(u_k^n(\hat{m}_k,\hat{l}_k),y_k^n) \in \mathcal{T}_\epsilon^{(n)}$.
	\end{enumerate}
	
	Similar to the probability of error analysis in the compression strategy, to
	show that the average probability of error $P_e^{(n)}$ for the coding scheme
	vanishes as $n \to \infty$, we analyze two sources of error. 
	For encoding at the CP, we can find the indices $[l_1:l_K,t_1:t_L]$ correctly
	with high probability if 
	$\sum_{k \in \mathcal{D}} \tilde{R}_k + \sum_{l \in \mathcal{S}} C_l>
	T(U(\mathcal{D}),X(\mathcal{S}))$, due to the multivariate covering lemma 
	\cite[Lemma 14.1]{ElGamalKim2011NetworkIT}. The decoding at the user side 
	is successful with high probability if
	$R_k + \tilde{R}_k < I(U_k;Y_k)$ due to the joint typicality lemma 
	\cite[p. 29]{ElGamalKim2011NetworkIT}. 
	Combining the two, we obtain the required rate region.

\section{Proof of Theorem \ref{thm:sum_fronthaul}}
\label{appendix:sum_fronthaul}

	We examine the set of achievable rate tuples $(R_1,\ldots,R_K)$ of the
	generalized compression and the DDF strategies under a sum fronthaul
	constraint $C$. Since the compression strategy is a special case of the DDF
	strategy, we have that $\mathcal{R}^{\mathrm{s}}_{\textrm{COM}}(C) \subseteq
	\mathcal{R}^{\mathrm{s}}_{\textrm{DDF}}(C)$. The main part of the proof is to
	show that $\mathcal{R}^{\mathrm{s}}_{\textrm{DDF}}(C) \subseteq
	\mathcal{R}^{\mathrm{s}}_{\textrm{COM}}(C)$. The proof uses properties of
	submodular optimization.
		
	Take any achievable $(R_1,\ldots,R_K)$ using the DDF strategy under a fixed distribution 
	$p(u_1,\ldots,u_K,x_1,\ldots,x_L)$ and under the fixed sum fronthaul 
	constraint $C$. By definition, it must satisfy the inequalities
	\eqref{eq:DDF_sumfronthaul_1} and \eqref{eq:DDF_sumfronthaul_2}.
	Define $\underline{\mathcal{P}}(C) \in \mathbb{R}^K$ to be the polytope formed by the inequalities
	\eqref{eq:DDF_sumfronthaul_1} and \eqref{eq:DDF_sumfronthaul_2}. We show that
	each extreme point of $\underline{\mathcal{P}}(C)$ can be achieved using
	the time-sharing of rate tuples under the generalized compression strategy. 

	The inequalities \eqref{eq:DDF_sumfronthaul_1}-\eqref{eq:DDF_sumfronthaul_2}
	define $\underline{\mathcal{P}}(C)$ to be set of $(R_1,\ldots,R_K)$ for which
	\begin{equation}
	\label{eq:tight_before_1}
	\sum_{k \in \mathcal{D}} R_k  \le 
	\min \left\{ \begin{aligned} & \sum_{k \in \mathcal{D}} I(U_k;Y_k) - T(U(\mathcal{D})), \\
	& \sum_{k \in \mathcal{D}} I(U_k;Y_k) + C - T(U(\mathcal{D}),X(\mathcal{L}))
				\end{aligned} \right\}
	\end{equation}
	for all $\mathcal{D} \subseteq \mathcal{K}$. First, we show that we can
	alternatively write the above as
	\begin{equation}
	\label{eq:tight_before_2}
	\sum_{k \in \mathcal{D}} R_k  \le 
	\min \left\{ \begin{aligned} & \sum_{k \in \mathcal{D}} I(U_k;Y_k) - T(U(\mathcal{D})), \\
	& \sum_{k \in \mathcal{K}} I(U_k;Y_k) + C - T(U(\mathcal{K}),X(\mathcal{L})) \end{aligned} \right\}
	\end{equation}
	for all $\mathcal{D} \subseteq \mathcal{K}$. 
	The reason is that for any set $\mathcal{D} \subseteq \mathcal{K}$, we always
	have $\sum_{k \in \mathcal{D}} R_k \le \sum_{k \in \mathcal{K}} R_k$, 
	since the user rates are non-negative. But we already have the constraint
	\begin{equation}
	\sum_{k \in \mathcal{K}} R_k \le \sum_{k \in \mathcal{K}} I(U_k;Y_k) + C -
	T(U(\mathcal{K}),X(\mathcal{L})).
	\end{equation} 
	So, we can add the constraint 
	\begin{equation}
	\sum_{k \in \mathcal{D}} R_k \le \sum_{k \in \mathcal{K}} I(U_k;Y_k) + C -
	T(U(\mathcal{K}),X(\mathcal{L}))
	\label{eq:over_K}
	\end{equation} 
	to (\ref{eq:tight_before_1}) without affecting $\underline{\mathcal{P}}(C)$. 
	Now, it turns out that 
	\begin{multline}
	\sum_{k \in \mathcal{K}} I(U_k;Y_k) + C - T(U(\mathcal{K}),X(\mathcal{L})) \le \\
	\sum_{k \in \mathcal{D}} I(U_k;Y_k) + C - T(U(\mathcal{D}),X(\mathcal{L})),
	\label{eq:tighter_constraint}
	\end{multline} 
	so this new constraint is actually tighter than the second constraint in
	(\ref{eq:tight_before_1}).  Therefore, (\ref{eq:tight_before_1}) can be
	equivalently written as (\ref{eq:tight_before_2}). 
	
	To verify (\ref{eq:tighter_constraint}), we take the difference in summing
	over $\mathcal{D}$ versus summing over $\mathcal{K}$ in (\ref{eq:tighter_constraint}) 
	as below: 
	\begin{align}
	&T(U(\mathcal{K}),X(\mathcal{L})) - T(U(\mathcal{D}),X(\mathcal{L})) - \sum_{k \in \mathcal{D}^c} I(U_k; Y_k) \nonumber \\	
	& \qquad = I(U(\mathcal{D}^c);U(\mathcal{D})) + T(U(\mathcal{D}^c)) \nonumber \\
	& \qquad \qquad + I(U(\mathcal{D}^c);X(\mathcal{L})|U(\mathcal{D})) 
	- \sum_{k \in \mathcal{D}^c} I(U_k; Y_k),	
	\end{align}
	where $\mathcal{D}^c = \mathcal{K} \setminus \mathcal{D}$. This can be simplified as
	\begin{align}
	& \sum_{k \in \mathcal{D}^c} h(U_k | Y_k) - h(U(\mathcal{D}^c) | X(\mathcal{L}), U(\mathcal{D})) \\
	& \quad \overset{(a)}\ge \sum_{k \in \mathcal{D}^c} h(U_k | Y_k) -  \sum_{k \in \mathcal{D}^c} h(U_k | X(\mathcal{L})) \\
	& \quad \overset{(b)}\ge \sum_{k \in \mathcal{D}^c} h(U_k | Y_k) -  \sum_{k \in \mathcal{D}^c} h(U_k | X(\mathcal{L}), Y_k) \\
	& \quad = \sum_{k \in \mathcal{D}^c} I(U_k, X(\mathcal{L}) | Y_k) \\
	& \quad \ge 0,
	\end{align}
	where $(a)$ follows from the fact that conditioning reduces entropy and $(b)$
	follows since $U_k \rightarrow X(\mathcal{L}) \rightarrow Y_k$ form a Markov
	chain. This verifies (\ref{eq:tighter_constraint}), hence the equivalence
	between (\ref{eq:tight_before_1}) and (\ref{eq:tight_before_2}). 
	
	Let us now define a set function $f : 2^\mathcal{K} \rightarrow \mathbb{R}$ as
	\begin{equation}
	f(\mathcal{D}) := \min \left\{ \begin{aligned} & \sum_{k \in \mathcal{D}} I(U_k;Y_k) - T(U(\mathcal{D})), \\
									& \sum_{k \in \mathcal{K}} I(U_k;Y_k) + C - T(U(\mathcal{K}),X(\mathcal{L}))
									\end{aligned} \right\} \label{eq:f_definition}
	\end{equation}
	for each $\mathcal{D} \subseteq \mathcal{K}$.
	By construction, $\underline{\mathcal{P}}(C)$ is the set of $(R_1,\ldots,R_K)$ that satisfies
	\begin{equation}
	\sum_{k \in \mathcal{D}} R_k \le f(\mathcal{D}).
	\end{equation}
	Since the second term in the $\min$ expression in
\eqref{eq:f_definition} is a constant that does not depend of $\mathcal{D}$, 
it can be verified that the function $f$ is a submodular function
\cite{zhou2016optimal}, if the Marton's region is a polymatroid (which we
assume in this paper). This allows the rate region $\underline{\mathcal{P}}(C)$
to have a polymatroid structure. We remark that, although the Marton's region
may not be polymatroid in general, for the constant gap Gaussian distribution,
we can guarantee a certain monotone property of the Marton's rate expression by
appropriate choice of the noise variance leading to a polymatroid rate region. 

A result in submodular optimization~\cite{Schrijver2003} is that for a linear
	ordering $i_1 \prec i_2 \prec \ldots \prec i_K$ of $\{1,\ldots,K\}$, an extreme
	point of $\underline{\mathcal{P}}(C)$ can be greedily computed as
	$({R}_1,\ldots,{R}_K)$ where
	\begin{equation}
	{R}_{i_j} = f(\{i_1,\ldots,i_j\}) - f(\{i_1,\ldots,i_{j-1}\}).
	\end{equation}
	Moreover, all extreme points of $\underline{\mathcal{P}}(C)$ can be enumerated by considering all linear orderings. Since each ordering of $\{1,\ldots,K\}$ is analyzed in the same manner, for notational simplicity, we consider the natural ordering $i_j = j$.
	
	Let $j$ be the first index for which 
	\begin{multline}
	\label{eq:assumption_sumfronthaul}
	\sum_{k \in \mathcal{K}} I(U_k;Y_k) + C - T(U(\mathcal{K}),X(\mathcal{L})) < \\
	\sum_{k = 1}^j I(U_k;Y_k) - T(U_1,\ldots,U_{j}).
	\end{multline} 
	Then, by construction, $\forall k < j$
	\begin{align}
	{R}_k &= \sum_{i=1}^k I(U_i;Y_i) - T(U_1,\ldots,U_k) \nonumber \\
	& \qquad - \sum_{i=1}^{k-1} I(U_i;Y_i) - T(U_1,\ldots,U_{k-1}) \\
	&= I(U_k;Y_k) - I(U_k;U_{k-1},\ldots,U_1).
	\label{eq:RR_1}
	\end{align}
	Furthermore, using the fact that the second term of $f(\mathcal{D})$ does not
	depend on $\mathcal{D}$, so when the second term is the minimum, i.e., 
	$\forall k > j$, we have
	\begin{equation}
	{R}_k = 0.
	\label{eq:RR_2}
	\end{equation}
	Finally, we express ${R}_j$ as
	\begin{align}
	{R}_j & = \sum_{k \in \mathcal{K}} I(U_k;Y_k) + C - T(U(\mathcal{K}),X(\mathcal{L})) \nonumber \\
	& \qquad - \sum_{k = 1}^{j-1} I(U_k;Y_k) - T(U_1,\ldots,U_{j-1}) \\
	&= I(U_j;Y_j) - I(U_j;U_{j-1},\ldots,U_1) \nonumber \\
	& \qquad + \sum_{k \in \mathcal{K}} I(U_k;Y_k) + C - T(U(\mathcal{K}),X(\mathcal{L})) \nonumber \\
	& \qquad \qquad - \sum_{k = 1}^{j} I(U_k;Y_k) + T(U_1,\ldots,U_{j}) \\		
	&= (1-\alpha) \left( I(U_j;Y_j) - I(U_j;U_{j-1},\ldots,U_1) \right),
	\label{eq:RR_3}
	\end{align}
	where $\alpha$ can be written explicitly as below 
	\begin{equation}
\frac{ T(U(\mathcal{K}),X(\mathcal{L})) - T(U_1,\ldots,U_{j}) - C - \displaystyle \sum_{k = {j+1}}^K I(U_k;Y_k)} {I(U_j;Y_j) - I(U_j;U_{j-1},\ldots,1)}.
	\end{equation}
	Note that $\alpha \in (0,1]$ due to \eqref{eq:assumption_sumfronthaul} 
	and the fact that ${R}_j \ge 0$.
	
	Now, we construct the time-sharing of two rate tuples corresponding to
	the generalized compression strategy (\ref{eq:COM_sumfronthaul}) that achieves
	this above rate $({R}_1,\ldots,{R}_K)$ as follows: 
	\begin{enumerate}
		\item For $(1-\alpha)$ fraction of the time, transmit messages for users $1,\ldots,j$,
		only, i.e., set $p(u_1,\ldots,u_j,x_1,\ldots,x_L)$ to be the marginal
		distribution of the original distribution, but let $U_{j+1},\dots,U_K$ be constants. 
		\item For the rest $\alpha$ fraction of time, transmit messages for users 
		$1,\ldots,({j-1})$ only, i.e., set $p(u_1,\ldots,u_{j-1},x_1,\ldots,x_L)$ to be 
		the marginal distribution of the original distribution, but let $U_{j},\dots,U_K$ be constants. 
	\end{enumerate}
	
	By construction, the time-sharing of these two compression schemes achieves the
	same rate tuple as the extreme point of $\underline{\mathcal{P}}(C)$,
	(\ref{eq:RR_1}), (\ref{eq:RR_2}), and (\ref{eq:RR_3}).
	
	To calculate the fronthaul capacity consumption of this time-sharing scheme, 
	we have
	\begin{align}
\bar{C} &= (1-\alpha) \left ( I(U_1,\ldots,U_j;X(\mathcal{L})) + T(X(\mathcal{L})) \right) \nonumber \\
& \qquad \alpha \left( I(U_1,\ldots,U_{j-1};X(\mathcal{L})) + T(X(\mathcal{L})) \right) \\	
&= I(U_1,\ldots,U_j;X(\mathcal{L})) + T(X(\mathcal{L})) \nonumber \\
& \qquad - \alpha(I(U_1,\ldots,U_j;X(\mathcal{L})) - I(U_1,\ldots,U_{j-1};X(\mathcal{L})))\\
&= I(U_1,\ldots,U_j;X(\mathcal{L})) + T(X(\mathcal{L})) - C \nonumber \\
& \qquad - \frac{I(U_1,\ldots,U_j;X(\mathcal{L})) - I(U_1,\ldots,U_{j-1};X(\mathcal{L}))}{I(U_j;Y_j) - I(U_j;U_{j-1},\ldots,U_1)} \nonumber \\
& \qquad \quad  
\cdot \left( \rule{0pt}{7mm} T(U(\mathcal{K}),X(\mathcal{L})) - T(U_1,\ldots,U_{j}) - C \right.
\nonumber \\
& \qquad \qquad \left. - \sum_{k = {j+1}}^K I(U_k;Y_k) \right) + C \\
&\overset{(a)}{\le} I(U_1,\ldots,U_j;X(\mathcal{L})) + T(X(\mathcal{L})) -C \nonumber \\
& \qquad -  T(U(\mathcal{K}),X(\mathcal{L})) + T(U_1,\ldots,U_{j}) +C \nonumber \\
& \qquad \quad + \sum_{k = {j+1}}^K I(U_k;Y_k) + C \label{eq:rate_fronthaul_comparison}\\
& = \left( \sum_{k \in \mathcal{K}} I(U_k;Y_k) + C - T(U(\mathcal{K}),X(\mathcal{L})) \right) \nonumber \\
& \qquad - \left( \sum_{k=1}^j I(U_k;Y_k) + C - T(U_1,\ldots,U_{j},X(\mathcal{L})) \right) \nonumber \\
& \qquad \quad + C \\
&\overset{(b)}{\le} C.
\end{align}
The inequality $(a)$ follows because 
\begin{align} 
& I(U_1,\ldots,U_j;X(\mathcal{L})) - I(U_1,\ldots,U_{j-1};X(\mathcal{L})) \nonumber \\
& \qquad - I(U_j;Y_j) + I(U_j;U_{j-1},\ldots,U_1) \\
& = h(U_j|Y_j) - h(U_j|X(\mathcal{L}), U_{j-1}, \ldots, U_1) \\
& \ge h(U_j|Y_j) - h(U_j|X(\mathcal{L})) \\
& = h(U_j|Y_j) - h(U_j|X(\mathcal{L}),Y_j) \\
& = I(U_j;X(\mathcal{L})|Y_j) \\
& \ge 0,
\end{align}
where we used the fact that conditioning reduces entropy and that $U_j \rightarrow X(\mathcal{L}) \rightarrow Y_j$ forms a Markov chain. Intuitively, this holds because the contribution of $U_j$ to the user rate is less than the fronthaul required to support $U_j$. Note that the term $T(U(\mathcal{K}),X(\mathcal{L})) - T(U_1,\ldots,U_{j}) - C - \sum_{k = {j+1}}^K I(U_k;Y_k)$ is positive from the assumption in \eqref{eq:assumption_sumfronthaul}. The inequality $(b)$ follows from (\ref{eq:tighter_constraint}). 
	
	Therefore, every extreme point $({R}_1,\ldots,{R}_K)$ of
	$\underline{\mathcal{P}}(C)$ is achievable using time-sharing of generalized 
	compression strategies under the same average fronthaul constraint.

\section{Proof of Theorem \ref{thm:sum_rate}}

\label{appendix:sum_rate}

The proof is based on comparing the sum rate achieved by the
compression strategy with that by the DDF strategy. 
Recall that, from the result in \cite{Lim2017DistributedDecodeForward}, for the
DDF strategy the following choice of the distribution achieves to within a
constant gap to the cut-set bound of a Gaussian relay broadcast network: Let
$\mathbf{X}$ to be a vector of $L$ i.i.d.\ $\mathcal{N}(0,P)$ random variables 
and $\mathbf{U} = \mathbf{H} \mathbf{X} + \tilde{\mathbf{Z}}$, where 
$\tilde{\mathbf{Z}} \sim \mathcal{N}(0,\sigma^2 \mathbf{I})$ is independent of 
$\mathbf{Z}$. We show that under such a choice of distribution
${R}^\text{g}_{\text{DDF}} =  {R}^\text{g}_{\text{COM}}$. 
Then it follows that the compression strategy also achieves the sum rate to
within a constant gap to the cut-set bound.  
	
Consider the set of $(R,C_1,\ldots,C_L)$ achievable using the DDF strategy
under such a constant-gap distribution.
For fixed $R$, we define $\overline{\mathcal{P}}(R) \subseteq \mathbb{R}^L$ to be the
polytope defined by inequalities \eqref{eq:sumrate_DDF_1} under the said
distribution.
We now show that each extreme point of $\overline{\mathcal{P}}(R)$ is dominated
by some time-sharing of points in the compression region.
	
	Let us define a set function $g : 2^\mathcal{L} \rightarrow \mathbb{R}$ as
	\begin{equation}
	g(\mathcal{S}) := \max \left\{ \begin{aligned} & T(U(\mathcal{K}),X(\mathcal{S})) + R - \sum_{k \in \mathcal{K}} I(U_k;Y_k), \\
	& 0 \end{aligned} \right\} \label{eq:g_definition}
	\end{equation}
	for each $\mathcal{S} \subseteq \mathcal{L}$.
	By construction, then $\overline{\mathcal{P}}(R)$ is equal to the set of $(C_1,\ldots,C_L)$ that satisfy
	\begin{equation}
	\sum_{l \in \mathcal{S}} C_l \ge g(\mathcal{S}).
	\end{equation}
	Since the second term in the $\max$ expression in
\eqref{eq:g_definition} is a constant, it can be verified that the function $g$
is a supermodular function \cite{courtade2014multiterminal} and as a
consequence the $\overline{\mathcal{P}}(R)$ region is a contra-polymatroid \cite{zhang2007successive}. Similar to the case of submodular optimization, for a linear ordering
$i_1 \prec i_2 \prec \ldots \prec i_K$ of $\{1,\ldots,K\}$, an extreme point of
$\overline{\mathcal{P}}(R)$ can be greedily computed as
	\begin{equation}
	{C}_{i_j} = g(\{i_1,\ldots,i_j\}) - g(\{i_1,\ldots,i_{j-1}\}).
	\end{equation}
	Furthermore, all the extreme points of $\overline{\mathcal{P}}(R)$ can be computed by considering all linear orderings.
	Each ordering of $\{1,\ldots,K\}$ is analyzed in the same manner, hence for notational simplicity we consider the natural ordering $i_j = j$.
	
	Let $j$ be the first index for which ${C}_j > 0$. Then, by construction,
	\begin{equation} {C}_l = 0, \qquad \forall l < j \end{equation}
and
	\begin{equation}
	\label{eq:sumrate_assumption}
	{C}_l = I(X_l;U(\mathcal{K})|X_{l-1},\ldots,X_1), \quad \forall l > j.
	\end{equation}
	Note that the term $T(X(\mathcal{S}))$ vanishes because of the
assumption of independence of $X$'s in the constant-gap distribution.
	Finally, we express ${C}_j$ as 
	\begin{align}
	{C}_j &= I(X_j,\ldots,X_1;U(\mathcal{K})) \nonumber \\
	& \qquad + R - \sum_{k \in \mathcal{K}} I(U_k;Y_k) + T(U(\mathcal{K})) \\
	&= I(X_j;U(\mathcal{K})|X_{j-1},\ldots,X_1) + I(X_{j-1},\ldots,X_1;U(\mathcal{K})) \nonumber \\
	& \qquad + R - \sum_{k \in \mathcal{K}} I(U_k;Y_k)+T(U(\mathcal{K}))
\label{eq:Cj_expression} \\
	&= (1-\beta) I(X_j;U(\mathcal{K})|X_{j-1},\ldots,X_1) \label{eq:C_j_relation},
	\end{align}
	where $\beta$ is defined as
	\begin{align}
\frac{-\left( {I(X_{j-1},\ldots,X_1;U(\mathcal{K}))}{+R-\sum_{k \in \mathcal{K}} I(U_k;Y_k)+T(U(\mathcal{K}))} \right)} {I(X_j;U(\mathcal{K})|X_{j-1},\ldots,X_1)}
	\end{align}
It is not difficult to see that $\beta \in (0,1]$.  This is because $j$ is the
first index for which ${C}_{j}>0$, so ${C}_{j-1}=0$. By definition
of ${C}_j$, it is easy to see that $g(\{1,\ldots,j-1\})=0$.  Observe that
the numerator in the expression for $\beta$ is the negative of the first term
in the definition of $g(\{1,\ldots,{j-1}\})$, so the numerator must be
positive, hence $\beta > 0$. Further, by (\ref{eq:Cj_expression}) and the
fact that ${C}_{j} \ge 0$, we have $\beta \le 1$. 

	Now, consider the following time-sharing of two compression schemes. 
Starting with the fixed constant-gap distribution $p(u_1,\ldots,u_K,x_1,\ldots,x_L)$,  
we modify the distribution as follows:
\begin{enumerate}
		\item For $(1-\beta)$ fraction of the time, keep the BSs
$j,\ldots,L$ active, i.e., for $(1-\beta)$ fraction of the time, keep 
$X_j,\ldots,X_L$ the same and set $X_1=\ldots=X_{j-1}=0$;
denote this distribution as $p(u'_1,\ldots,u'_K,x'_1,\dots,x'_L)$.
		\item For the remaining $\beta$ fraction of the time, keep the
BSs $j+1,\ldots,L$ active, i.e., for $\beta$ fraction of the time, keep
$X_{j+1},\ldots,X_L$ the same and set $X_1=\ldots,X_{j}=0$; 
denote this distribution as $p(u''_1,\ldots,u''_K,x''_1,\dots,x''_L)$.
\end{enumerate}

We first verify that the average fronthaul capacities required for this
time-sharing of two compression schemes, denoted here as 
$\bar{C}_1,\ldots,\bar{C}_L$, are exactly the same as the fronthaul capacities
${C}_1,\ldots,{C}_L$ under the DDF strategy.
For the inactive BSs from 1 to ${j-1}$ the fronthaul
capacities used is zero, i.e.,
\begin{equation}
\bar{C}_l = 0 = {C}_l, \qquad \forall l=1,\ldots,{j-1}. 
\end{equation}
We use the modified distributions under the compression strategy 
to calculate the fronthaul needed for the active BSs.
Note that under the constant-gap distribution (or its modified
form), a corner point of the fronthaul region (\ref{eq:sumrate_COM_2}) is just
\begin{equation}
C_l = I(X_l; U(\mathcal{K}) | X_{l-1},\ldots,X_1) 
\end{equation}
where the term $T(X(\mathcal{S}))$ vanishes because of the assumed independence
of $X$'s in the constant-gap distribution.

Now for BS $j$, since $X_1'=\ldots=X'_{j-1}=0$, the fronthaul used by the
compression strategy is just
\begin{align}
\bar{C}_j & = (1-\beta) I(X'_j;U'(\mathcal{K})) \\ 
	& = (1-\beta)I(X_j;U(\mathcal{K})|X_{j-1},\ldots,X_1) = {C}_j,
\end{align}
where the equality is due to the form of the Gaussian
$p(u_1,\ldots,u_K,x_1,\ldots,x_L)$ in which conditioning on $X_{j-1},\ldots,
X_1$ is the same as setting them to be zero.

For BS $l = (j+1), \ldots, L$, the fronthaul capacity used by the generalized 
compression strategy is given by 
\begin{align}
\bar{C}_j & = (1-\beta)I(X'_l;U'(\mathcal{K})|X'_{l-1},\ldots,X'_{j}) \nonumber \\
& \qquad + \beta I(X''_l;U''(\mathcal{K})|X''_{l-1},\ldots,X''_{j+1}) \\
& = (1-\beta)I(X_l;U(\mathcal{K})|X_{l-1},\ldots,X_{j},X_{j-1},\ldots,X_1) \nonumber \\
& \qquad + \beta I(X_l;U(\mathcal{K})|X_{l-1},\ldots,X_{j+1},X_j,\ldots,X_1) \\
& = I(X_l;U(\mathcal{K})|X_{l-1},\ldots,X_1) = {C}_j,	
\end{align}
This verifies that the time-sharing strategy uses the same amount of fronthaul as
DDF. 
	
	As a final step, we show that the time-sharing of the two compression
schemes achieves a sum rate no less than the DDF strategy.
First, we re-write the sum rate expression under the constant-gap distribution
(or its modified version) in a form that shows explicit dependence on the $X$
variables. 

\begin{lemma} 
	\label{lemma:equi_sumrate}
	Suppose that $\mathbf{X}$ is a vector of indepedent variables, and $\mathbf{Y}=\mathbf{H}\mathbf{X}+\mathbf{Z}$, $\mathbf{Y}=\mathbf{G}\mathbf{Y}+\mathbf{\tilde{Z}}$,
	where $\mathbf{H}$ and $\mathbf{G}$ are fixed matrices and $\mathbf{Z}$ and
	$\mathbf{\tilde{Z}}$ are vectors of independent variables that are also independent
	of each other and of $\mathbf{X}$. Then, 
	\begin{multline}
	\sum_{k \in \mathcal{K}} I(U_k;Y_k) - T(U(\mathcal{K})) = \\
	I(U(\mathcal{K});X(\mathcal{L})) - \sum_{k \in \mathcal{K}} I(U_k;X(\mathcal{L})|Y_k) 
	\label{eq:sumrate_equivalence} 
	\end{multline}
\end{lemma}

\begin{proof}
	\begin{align}
	& \sum_{k \in \mathcal{K}} I(U_k;Y_k) - T(U(\mathcal{K})) \\
	& \quad = h(U(\mathcal{K})) - \sum_{k \in \mathcal{K}} h(U_k|Y_k) \\ 
	& \quad = h(U(\mathcal{K})) - h(U(\mathcal{K})|X(\mathcal{L})) + h(U(\mathcal{K})|X(\mathcal{L})) \nonumber \\
	&  \quad \qquad - \sum_{k \in \mathcal{K}} h(U_k|Y_k) \\
	& \quad \overset{(a)}{=}  I(U(\mathcal{K});X(\mathcal{L})) + \sum_{k \in \mathcal{K}} \left( h(U_k|X(\mathcal{L}),Y_k) - h(U_k|Y_k) \right) \\
	& \quad = I(U(\mathcal{K});X(\mathcal{L})) - \sum_{k \in \mathcal{K}} I(U_k;X(\mathcal{L})|Y_k),
	\end{align} 
	where in $(a)$ we used the fact that $U(\mathcal{K}) \rightarrow X(\mathcal{K})
	\rightarrow Y(\mathcal{K})$ forms a Markov chain and that conditioned on
	$X(\mathcal{L})$ the $U$'s are independent.
\end{proof}
	
Based on the Lemma, 
the sum rate achieved using the time-sharing of the two generalized
compression schemes with the modified constant-gap distributions can be written as
\begin{align}
\bar{R} &= (1-\beta) \big( I(X'_j,\ldots,X'_L;U'(\mathcal{K})) \nonumber \\
&\qquad - \sum_{k \in \mathcal{K}} I(U'_k;X'_j,\ldots,X'_L|Y'_k) \big) \nonumber \\
& \qquad \quad + \beta \big( I(X''_{j+1},\ldots,X''_L;U''(\mathcal{K})) \\ \nonumber
& \qquad \quad\quad  - \sum_{k \in \mathcal{K}} I(U''_k;X''_{j+1},\ldots,X''_L|Y''_k) \big) \\
&\overset{(a)}{=} (1-\beta) I(X_j,\ldots,X_L;U(\mathcal{K})|X_{j-1},\ldots,X_1) \nonumber \\
& \qquad + \beta I(X_{j+1},\ldots,X_L;U(\mathcal{K})|X_j,\ldots,X_1) \nonumber \\
& \qquad \quad - (1-\beta) \sum_{k \in \mathcal{K}} I(U'_k;X'_j,\ldots,X'_L|Y'_k) \nonumber \\
&\qquad \quad\quad  - \beta \sum_{k \in \mathcal{K}} I(U''_k;X''_{j+1},\ldots,X''_L|Y''_k) \\
&= (1-\beta)  \big (I(U(\mathcal{K});X_j|X_{j-1},\ldots,X_1) \nonumber \\
& \qquad + I(U(\mathcal{K});X_{j+1},\ldots,X_L|X_j,\ldots,X_1) \big) \nonumber \\
& \qquad \quad + \beta I(U(\mathcal{K});X_{j+1},\ldots,X_L|X_{j},\ldots,X_1) \nonumber \\
& \qquad \quad \quad - (1-\beta) \sum_{k \in \mathcal{K}} I(U'_k;X'_j,\ldots,X'_L|Y'_k) \nonumber \\
&\qquad \quad \quad \quad - \beta \sum_{k \in \mathcal{K}} I(U''_k;X''_{j+1},\ldots,X''_L|Y''_k) \\
&\overset{(b)}{=} I(U(\mathcal{K});X_1,\ldots,X_L) + R - \sum_{k \in \mathcal{K}} I(U_k;Y_k) + T(U(\mathcal{K})) \nonumber \\
& \qquad \quad - (1-\beta) \sum_{k \in \mathcal{K}} I(U'_k;X'_j,\ldots,X'_L|Y'_k) \nonumber \\
& \qquad \quad \quad - \beta \sum_{k \in \mathcal{K}} I(U''_k;X''_{j+1},\ldots,X''_L|Y''_k) \\
&\overset{(c)}{\ge} I(U(\mathcal{K});X_1,\ldots,X_L) + R - \sum_{k \in \mathcal{K}} I(U_k;Y_k) + T(U(\mathcal{K})) \nonumber \\
& \qquad -  \sum_{k \in \mathcal{K}} I(U_k;X(\mathcal{L})|Y_k) \label{eq:conditioning_inequality} \\
&\overset{(d)}{=} R \label{eq:rate_inequality}.
\end{align}
The equality $(a)$ holds, because as mentioned before, for the constant-gap
distribution, shutting down a BS is the same as conditioning on the
corresponding random variable. For the first $(1-\beta)$ fraction of time, 
we condition on $X_1,\ldots,X_{j-1}$, and for the rest $\beta$ fraction of
the time, we condition on $X_1,\ldots,X_{j}$.  The equality $(b)$
holds from the relation \eqref{eq:C_j_relation}.
The inequality $(c)$ follows because under the modified constant-gap
distribution, 
\begin{align}
I(U'_k; & X'_j,\ldots,X'_L|Y'_k) \nonumber \\
&= h(U'_k|Y'_k) - h(U'_k|X'_j,\ldots,X'_L,Y'_k) \\
&= h(U'_k,Y'_k) - h(Y'_k) - h(\tilde{Z}_k) \\
&= \frac{1}{2}\log\left(1 + \frac{\sum_{l=j}^L h_{k,l}^2P}{\sum_{l=j}^L h_{k,l}^2P + \sigma^2}\right)  \\
&< \frac{1}{2}\log\left(1 + \frac{\sum_{l=1}^L h_{k,l}^2P}{\sum_{l=1}^L h_{k,l}^2P + \sigma^2}\right)  \\
&= I(U_k;X(\mathcal{L})|Y_k),
\end{align}
and similarly $I(U''_k;X''_{j+1},\ldots,X''_L|Y''_k) < I(U_k;X(\mathcal{L})|Y_k)$.
Finally, the equality $(d)$ follows from the equivalent way of
writing the sum rate as shown in Lemma \ref{lemma:equi_sumrate}.
	
Therefore, for every extreme point $({C}_1,\ldots,{C}_L)$ of
$\overline{\mathcal{P}}(R)$, the time-shared compression strategy achieves a
sum rate at least as large as the DDF strategy. This completes the proof.

\end{appendices}

\bibliographystyle{IEEEtran}

\bibliography{IEEEabrv,confabrv,references}

\end{document}